\documentclass[onefignum,onetabnum]{siamonline171218}

\usepackage{graphicx}
\usepackage{amsmath,amssymb,amsfonts}
\usepackage{mathrsfs,mathtools}
\usepackage{subfig}
\usepackage[shortlabels]{enumitem}
\usepackage{algorithmic}
\usepackage[algo2e]{algorithm2e}
\usepackage{bm,bbm}

\newsiamremark{remark}{Remark}
\newsiamremark{hypothesis}{Hypothesis}
\crefname{hypothesis}{Hypothesis}{Hypotheses}
\newsiamthm{claim}{Claim}

\newtheorem{assumption}{Assumption}

\headers{Adaptive particle-based approximations of the Gibbs posterior for inverse problems}{Z. Zou}

\title{Adaptive particle-based approximations of the Gibbs posterior for inverse problems\thanks{
\funding{HA was partially supported by NSF grants DMS-1818772 
and DMS-1913004 and Air Force Office of Scientific Research under Award 
NO: FA9550-19-1-0036. WA and ZZ were partially supported by DARPA EQUiPS program, grant SNL 014150709. SM was partially supported
by NSF grants DMS 1613261, DEB-1840223, DMS 1713012 and HFSP RGP0051/2017.}}}

\author{Zilong Zou\thanks{Department of Civil and Environmental Engineering, Duke University, NC 
  (\email{zilong.zou@duke.edu})}
\and
Sayan Mukherjee\thanks{Departments of Statistical Science, Mathematics, Computer Science, Biostatistics \& Bioinformatics, Duke University, NC 
	(\email{sayan@stat.duke.edu}).}
\and
Harbir Antil\thanks{Department of Mathematical Sciences, George Mason University, VA 
  (\email{hantil@gmu.edu}).}
\and
  Wilkins Aquino\thanks{Department of Mechanical Engineering and Materials Science, Duke University, NC 
	(\email{wilkins.aquino@duke.edu}).} 
}

\begin{document}
\pagenumbering{arabic}

\maketitle
\begin{abstract}
In this work, we adopt a general framework based on the Gibbs posterior to update belief distributions for inverse problems governed by partial differential equations (PDEs). 
The Gibbs posterior formulation is a generalization of standard Bayesian inference that only relies on a loss function connecting the unknown parameters to the data. 
It is particularly useful when the true data generating mechanism (or noise distribution) is unknown or difficult to specify. 
The Gibbs posterior coincides with Bayesian updating when a true likelihood function is known and the loss function corresponds to the negative log-likelihood,  yet provides subjective inference in more general settings.   \\

We employ a sequential Monte Carlo (SMC) approach to approximate the Gibbs posterior using particles. 

To manage the computational cost of propagating increasing numbers of particles through the loss function, 
we employ a recently developed local reduced basis method to build an efficient surrogate loss function that is used in the Gibbs update formula in place of the true loss. 
We derive error bounds for our approximation and propose an adaptive approach to construct the surrogate model in an efficient manner. 
We demonstrate the efficiency of our approach through several numerical examples. 
\end{abstract}



\section{Introduction}

In Bayesian inverse problems, we need to infer some unknown system parameters from the noisy measurement of the system response. 
Such problems are ubiquitous in many application areas including medical imaging \cite{franck2016sparse, koutsourelakis2012novel}, heat conduction \cite{wang2004bayesian}, 
geosciences \cite{bui2013computational}, atmospheric and oceanic sciences \cite{bennett2005inverse}. 
The Bayesian approach has been a foundation for performing such inference from noisy or incomplete observations while allowing us to quantify the uncertainty of the inverse solution due to the inexactness of data \cite{cotter2010approximation, stuart2010inverse}.
The solution of a Bayesian inverse problem is a probability distribution over the parameter space, which is referred to as the posterior distribution. 
Except for very limited settings, e.g., a linear model with Gaussian prior and Gaussian noise, the analytical form of the posterior distribution is rarely tractable. 
In most cases, we can only approximate the posterior either through sampling or parametrization. 

The main contribution of this work is a novel method to solve Bayesian inverse problems, which employs adaptive reduced bases to efficiently explore the posterior at a fraction of the cost of existing methods.  Moreover, we develop this approach in the context of the Gibbs posterior; a concept that offers very attractive features for inverse problems (e.g. use of a loss function instead of a likelihood).

Markov chain Monte Carlo (MCMC) \cite{gamerman2006markov} is the most well-known and versatile method for Bayesian inverse problems \cite{stuart2010inverse}. 
MCMC requires only pointwise evaluation of the likelihood function to generate a stream of samples from the posterior distribution that can be subsequently used to compute the statistics of the posterior. 
In MCMC, however, the generation of each new sample requires multiple evaluations of the likelihood function, which can be extremely expensive depending on the underlying model for the system response.
In this work, for example, we assume the system is modeled by PDEs that are expensive to evaluate. 

To accelerate MCMC, one popular approach is to employ an inexpensive surrogate model to approximate the likelihood evaluation in the sampling procedure. 
Plenty of research efforts have been devoted to construct efficient surrogate models for such a purpose. 
For example, the stochastic spectral method is used in \cite{marzouk2007stochastic}, Gaussian process regression is used in \cite{kennedy2001bayesian} and projection-based model reduction is employed in \cite{cui2015data, manzoni2016accurate}. 
The surrogate models are typically constructed to be accurate over the support of the prior distribution 
\cite{frangos2010surrogate, manzoni2016accurate, marzouk2009stochastic, marzouk2009dimensionality, marzouk2007stochastic} and are thought to be ``globally accurate".
However, thanks to the information contained in the data, the posterior distribution typically concentrate on a much                              smaller portion of the support of the prior.
In this respect, requiring a ``globally accurate" surrogate model seems unnecessary and inefficient. 
Several recent studies have exploited such information (or posterior) and build adaptive 
and data-driven surrogate models that are more efficient and accurate on the support of the posterior \cite{ conrad2016accelerating, conrad2014accelerating, cui2015data, li2014adaptive}.
Significant computational savings can be realized through such adaptive methods. 

Recently, sequential Monte Carlo (SMC) methods, or particle filters \cite{bernardo2007sequential, del2006sequential, doucet2001introduction}, have been applied in the setting of Bayesian inverse problems by a few researchers \cite{beskos2015sequential, kantas2014sequential}.
In SMC, weighted samples, or particles, are generated and evolved to approximate a sequence of probability distributions which interpolate from the prior to the posterior. 
In \cite{kantas2014sequential} in particular, the authors employed a novel SMC method with a dimension-independent MCMC sampler \cite{cotter2013mcmc, law2014proposals} as the mutation kernel to invert the initial conditions for Navier-Stokes equations.
In \cite{beskos2015sequential}, the authors enhanced the SMC method in \cite{kantas2014sequential} and provided a proof of the dimension-independent convergence property of the SMC methods for inverse problems. 
Both works have demonstrated the computational efficiency of SMC methods for high-dimensional inverse problems. 
In addition, the versatility and self-adaptivity of SMC methods provide a natural framework for the adaptive construction of surrogate models that can be used to further speedup the computations for inverse problems. 

The majority of Bayesian methods for inverse problems rely on an exact noise model, typically assumed to be i.i.d. Gaussian, to perform inference. 
It is desirable to extend such inference to more general settings where a noise model is unavailable or modeling the data generating mechanism is challenging. 
The Gibbs posterior provides a way to update belief distributions in such general setting without the need of an explicit likelihood function. 
Instead, the Gibbs posterior are applicable where the unknown parameters are only connected to the data through a loss function \cite{alquier2016properties, bissiri2016general, syring2016robust}. 
In many inverse problems or inference problems, it can be a simpler task to specify a loss function than the true data generating mechanism, i.e., an explicit likelihood function, 
which is the biggest advantage of using the Gibbs posterior over the usual Bayesian approach. 

In this work, we employ a particle-based approach to approximate the Gibbs posterior for inverse problems.
To manage the computational cost of propagating increasing number of particles through the loss function, 
we employ a recently developed local reduced basis method (local RB) \cite{zou2017adaptive, zou2018adaptive} to build an efficient surrogate model for the loss function that is used in the Gibbs update formula in place of the true loss. 
Based on a sequential Monte Carlo (SMC) framework developed in \cite{kantas2014sequential}, 
we present a method to progressive approximate the Gibbs posterior by simultaneously evolving the particles and adapting the local RB surrogate model in a sequential manner.
The emphasis of the local RB surrogate is navigated to a small fraction of the parameter space, i.e., the support of the posterior, automatically by the evolving particles that progressively cluster over the support of the posterior.
Computational savings are achieved thanks to the local accuracy and the efficiency of the local RB method \cite{zou2018adaptive}. 
Indeed, once the local RB surrogate becomes accurate enough (specified by a parameter representing the approximation accuracy) over the local support of the posterior, further evolution of the particles takes minimal cost.
We derive error bounds for our approximation and demonstrate the computational efficiency of our approach through several numerical examples including advection-diffusion problems and elasticity imaging problems.

\section{Gibbs Posterior for inverse problems}\label{sec:gibbsPos}
We assume we have a system (governed by a PDE) with unknown deterministic parameters $\xi^*\in\Xi\subseteq \mathbb{R}^M$. For the sake of simplicity, we will work in finite dimensional spaces arising from a discretization of the underlying PDE model.  We further assume that we can observe noisy output $d\in\mathbb{R}^D$ from the system as
$$d = \mathcal{F}(\xi^*) + \epsilon,$$
where $ \mathcal{F}: \Xi \rightarrow \mathbb{R}^D$ is a model representing the system that maps each parameter to an observation. In our case, each evaluation of the map requires a numerical solution to a PDE. 
Moreover, $\epsilon\in\mathbb{R}^D$ is the random noise.
We define $d^* = \mathcal{F}(\xi^*)$ as the true data without noise.
We further assume we have a prior belief about $\xi^*$, which can be expressed as a prior distribution $\rho_0(\cdot):\Xi\rightarrow\mathbb{R}$. 
Without imposing strong assumptions or an explicit distribution model on $\epsilon$, we would like to integrate the information contained in the data into our belief about $\xi^*$. 

In a Gibbs posterior formulation, we do not need a likelihood function, instead, we have access to a given loss function $l(\cdot, \cdot):\Xi\times\mathbb{R}^D\rightarrow\mathbb{R}$. 
A loss function measures the discrepancy between the predicted output at a parameter to the  observation. For example, we could use
$$l(\xi, d) = \|\mathcal{F}(\xi)-d\|_{l_2}^2$$
as the loss function. 
Unlike a likelihood function that requires exact knowledge of the data generating mechanism (or noise model), loss functions are typically easier to specify for inverse problems, 
which is the biggest advantage of Gibbs posterior. 

When a set of observations $d_i,\ i=1,2,\dots, n$ are given,  we update our belief according to the following optimization problem
\begin{equation}\label{eq:gibbsOpt}
\rho(\xi)=\underset{\hat{\rho}\in P}{\text{argmin}} \ \int_{\Xi} W \sum_{i=1}^nl(\xi, d_i)\hat{\rho}(\xi)d\xi+D_{KL}(\hat{\rho}\| \rho_{0}).
\end{equation}
where $W$ is a weight for the loss that is yet to be specified. For now, we assume $W$ is a fixed positive constant and will describe possible methods to prescribe $W$ later on. Furthermore,
$D_{KL}(\rho\| \rho_{0})$ is the Kullback-Leibler (KL) divergence between the posterior and prior distributions and 
$P$ is the space of candidate posterior distributions of $\xi$. 
If we allow $P$ to contain all distributions over $\Xi$, we have an explicit update formula for $\rho(\xi)$ \cite{bissiri2016general} as
\begin{equation}\label{eq:gibbsPosFormula}
\rho(\xi)= \frac{\exp(-W\sum_{i=1}^nl(\xi, d_i))\rho_{0}(\xi)}{\int_{\Xi} \exp(-W\sum_{i=1}^nl(\xi, d_i))\rho_{0}(\xi)d\xi}.
\end{equation}
This is a coherent update formula in the sense that the use of sequential  data in a sequential manner yields the same distribution as if we used all the data simultaneously as in \eqref{eq:gibbsPosFormula}.
In addition, we can see from \eqref{eq:gibbsOpt} that the Gibbs posterior minimizes the expected loss with an added requirement that this posterior be close to the prior in the sense of the KL divergence. Also, notice that $W$ weights our relative belief between the information provided by the data versus the information in the prior.

Also, we can see that the usual Bayes rule is a special case of \eqref{eq:gibbsPosFormula} by using the negative log-likelihood as the loss function with $W=1$. 
Indeed, if we use $l(\xi, d_i) =-\log\left(\pi(d_i|\xi)\right)$ where $\pi(d_i|\xi)$ is the likelihood function, we get
$$\rho(\xi)= \frac{\prod_{i=1}^n\pi(d_i|\xi)\rho_{0}(\xi)}{\int_{\Xi} \prod_{i=1}^n \pi(d_i|\xi)\rho_{0}(\xi)d\xi}$$
which is the conventional Bayes rule.

Integrating more data points into the Gibbs update requires just a summation of the individual losses to form a cumulative loss $l(\xi, \{d_i\}_{i=1}^n) := \sum_{i=1}^nl(\xi, d_i)$. So, the dependence of the loss function on data is straightforward. Hence, without loss of generality, we denote a generic loss function $l(\xi)$ in the sequel for the sake of notation simplicity. In this case, the Gibbs update formula becomes
\begin{equation}\label{eq:gibbsPosFormula_L}
\rho(\xi)= \frac{\exp(-Wl(\xi))\rho_{0}(\xi)}{\int_{\Xi} \exp(-Wl(\xi))\rho_{0}(\xi)d\xi}.
\end{equation}

\section{Surrogate approximation}
In the update formula \eqref{eq:gibbsPosFormula_L}, evaluating $l(\xi)$ at each parameter $\xi$ requires an evaluation of a potentially expensive PDE model. 
In this work, we use a computationally inexpensive surrogate model $\overline{l}(\xi)$ to approximate the loss function $l(\xi)$. Although our exposition is generally applicable to any method used for building surrogate models, we will focus later on the integration of an adaptive local reduced basis approach \cite{zou2018adaptive} with the current Gibbs posterior framework.
 
Using $\overline{l}(\xi)$ for the Gibbs update in \eqref{eq:gibbsPosFormula_L} results in an approximate Gibbs posterior that can be sampled (approximated) at a low computational cost. 
However, it is important to understand the error in such an approximation in order to build an effective surrogate model with controlled accuracy. 
To this end, we first define the approximate Gibbs posterior $\overline{\rho}(\xi)$ as
\begin{equation}\label{eq:gibbsPosFormula_L_RB}
\overline{\rho}(\xi)= \frac{\exp(-W\overline{l}(\xi))\rho_{0}(\xi)}{\int_{\Xi} \exp(-W\overline{l}((\xi))\rho_{0}(\xi)d\xi}.
\end{equation}

To quantify the error introduced by using the surrogate $\overline{l}(\xi)$ in \eqref{eq:gibbsPosFormula_L_RB}, 
we derive a bound for the discrepancy between the approximate posterior $\overline{\rho}(\xi)$ and $\rho(\xi)$. 
For this purpose, we first state the following boundedness assumption on the loss function $l(\xi)$ and its surrogate $\overline{l}(\xi)$.
\begin{assumption}\label{as:gibbsLoss}
The loss functions $l(\xi)$ and $\overline{l}(\xi)$ are nonnegative and are uniformly bounded from above:
$\exists\ C_l, C_{\overline{l}} > 0$ independent of $\xi\in\Xi$ such that for all $\xi\in\Xi$
\begin{equation}
0 \le l(\xi) \le C_l, \ 0 \le \overline{l}(\xi) \le C_{\overline{l}}.
\end{equation}
\end{assumption}

To measure the distance between probability distributions, we use the  metric
$$h(\rho_1, \rho_2) = \sup_{|f|_{\infty}\le 1}\sqrt{\mathbb{E}|\rho_1[f]-\rho_2[f]|^2}, $$
where $\rho_1, \rho_2\in P$ are two possibly random elements in $P$,
the supremum is over all $f:\Xi\rightarrow\mathbb{R}$ such that $\sup_{\xi\in\Xi} |f(\xi)|\le 1$, and $\rho[f]=\int_{\Xi}f(\xi)\rho(\xi)d\xi$. 
The expectation is with respect to the randomness of $\rho_1, \rho_2$. 
In case where $\rho_1$ is determined, and $\rho_2$ is an approximation to $\rho_1$ through a randomized algorithm, e.g., Monte Carlo, the expectation is with respect to the randomness of the algorithm.
Note that $h$ is indeed a metric on $P$, in particular, it satisfies the triangle inequality \cite{beskos2015sequential, rebeschini2015can}. 

In addition, we define an $e$-feasible set as
\begin{equation}\label{eq:gibbs_e_feasible}
\Xi_e :=\{ \xi\in \Xi:  |l(\xi) - \overline{l}(\xi)|\le e \}
\end{equation}
where $e>0$ is some constant indicating the accuracy of the surrogate model $\overline{l}(\xi)$. 
We always assume that $e$ is small, e.g., $We\ll 1$. 
The set $\Xi_e$ contains all the parameters where the surrogate is accurate in the sense that the absolute difference between $l(\xi)$ and the $\overline{l}(\xi)$ is bounded by $e$. 
The complement of $\Xi_e$ is denoted by $\Xi_e^{\perp}:=\Xi\setminus\Xi_e$. 
Now, we state the following theorem regarding the accuracy of $\overline{\rho}(\xi)$, 
\begin{theorem}\label{thm:gibbsBdd_continuous}
Under Assumption \ref{as:gibbsLoss}, the following bound holds:
\begin{align}
    h(\rho, \overline{\rho}) \le 2\exp(WC_l) \,C \, W e + 2\exp(WC_l + W\max\{C_l, C_{\overline{l}}\}) \rho[\mathbbm{1}_{\Xi_e^{\perp}}] ,\nonumber
  \end{align}
  for some constants $C > 0$. 
\end{theorem}
\begin{proof}
First, note that for $\forall\ \xi\in\Xi_e$, we have that $-e\le l(\xi)-\overline{l}(\xi)\le e$. For $e$ sufficiently small, e.g., $We\ll 1$, we have
\begin{align}
|\exp(-Wl(\xi))-\exp(-W\overline{l}(\xi))| &\le \exp(-Wl(\xi)) |1-\exp(Wl(\xi)-W\overline{l}(\xi))| \le C W e\nonumber
\end{align}
for some $C > 0$. Let
\begin{align}\label{eq:Z1Z2}
Z_1= \int_{\Xi} \exp(-Wl(\xi))\rho_{0}(\xi)d\xi, \quad 
Z_2= \int_{\Xi} \exp(-W\overline{l}(\xi))\rho_{0}(\xi)d\xi.
\end{align}
Using Assumption \ref{as:gibbsLoss} and the fact that $|f|_{\infty}\le 1$, we have
\begin{align}
Z_1 &=  \int_{\Xi} \exp(-Wl(\xi))\rho_{1}(\xi)d\xi \ge \exp(-W C_l), \nonumber\\
Z_2 &\ge \int_{\Xi} \exp(-W\overline{l}(\xi))\rho_{2}(\xi)|f(\xi)|d\xi. \nonumber
\end{align}
In addition, we have 
\begin{align}
|Z_1-Z_2| &\le \int_{\Xi_e} |\exp(-Wl(\xi))-\exp(-W\overline{l}(\xi))| \rho_{0}(\xi)d\xi \nonumber\\
&+ \int_{\Xi_e^{\perp}}  |1-\exp(Wl(\xi)-W\overline{l}(\xi))|  \exp(-Wl(\xi)) \rho_{0}(\xi)d\xi  \nonumber\\
&\le C W e +  \int_{\Xi_e^{\perp}}  |1-\exp(Wl(\xi)-W\overline{l}(\xi))|  Z_1 \rho(\xi) d\xi \nonumber\\
&\le C W e + \exp(W\max\{C_l, C_{\overline{l}}\}) \rho[\mathbbm{1}_{\Xi_e^{\perp}}]. \nonumber
\end{align}
Hence
\begin{align}
|\rho [f]-\overline{\rho} [f]| &=  \left|\frac{\int_{\Xi}\exp(-Wl(\xi))\rho_0(\xi)f(\xi)d\xi}{\int_{\Xi} \exp(-Wl(\xi))\rho_0(\xi)d\xi}-\frac{\int_{\Xi}\exp(-W\overline{l}(\xi))\rho_0(\xi)f(\xi)d\xi}{\int_{\Xi} \exp(-W\overline{l}(\xi))\rho_0(\xi)d\xi}\right|\nonumber\\
&\le \frac{\int_{\Xi} |\exp(-Wl(\xi)) - \exp(-W\overline{l}(\xi)) | \rho_{0}(\xi)|f(\xi)|d\xi}{Z_1}\nonumber\\
&+\frac{|Z_2-Z_1|\int_{\Xi} \exp(-W\overline{l}(\xi))\rho_{2}(\xi)|f(\xi)|d\xi}{Z_1Z_2}\nonumber\\
&\le \frac{C W e + \exp(W\max\{C_l, C_{\overline{l}}\}) \rho[\mathbbm{1}_{\Xi_e^{\perp}}] }{Z_1} + \frac{|Z_2-Z_1|}{Z_1}\nonumber\\
&\le 2\exp(WC_l) C W e + 2\exp(WC_l + W\max\{C_l, C_{\overline{l}}\}) \rho[\mathbbm{1}_{\Xi_e^{\perp}}] .\nonumber
\end{align}
This completes the proof. 
\end{proof}
Note that $\rho[\mathbbm{1}_{\Xi_e^{\perp}}]  = \int_{\Xi_e^{\perp}}\rho(\xi)d\xi$ is exactly the posterior measure of $\Xi_e^{\perp}$. 
Theorem \ref{thm:gibbsBdd_continuous} says that, given a prescribed $e$, if the posterior measure of the region where the surrogate model $\overline{l}(\xi)$ is inaccurate, i.e., $\rho[\mathbbm{1}_{\Xi_e^{\perp}}]$, is small, 
the approximate posterior $\overline{\rho}$ is close to the true posterior $\rho$ (depending on the prescribed accuracy $e$). 
This indicates that the local RB surrogate model only needs to be accurate over the ``important region" where the majority of the posterior mass is contained. 

Indeed, thanks to the information contained in the data, the posterior distribution typically concentrate on a much smaller portion of the prior support. Hence, we usually do not need a 
 ``globally accurate" surrogate model over the entire support of the prior. 
The local RB method, discussed next, is naturally tailored to provide locally accurate approximations as shown in \cite{zou2017adaptive, zou2018adaptive}. 

\section{The local RB surrogate}
To construct the surrogate model $\overline{l}(\xi)$,  we employ the local RB method first introduced in \cite{zou2018adaptive}.  
We briefly describe the local RB method in this section. 
To this end, we consider an abstract variational problem described as: find $u(\xi):\Xi\rightarrow U$ such that
\[
  \langle M(u(\xi); \xi), v\rangle_{V^*,V} = 0
    \quad\forall\,v\in V, \quad\forall\,\xi\in\Xi.
\]
where $U$ is a Hilbert space containing the solution of the PDE, $V$ is a Hilbert space containing test functions, $V^*$ is the dual space of $V$, and $M(\cdot; \xi): U\rightarrow V^*$ is a bounded linear operator for all $\xi\in\Xi$.

We assume the loss function $l(\xi) = g(u(\xi))$ for some functional $g:U\rightarrow \mathbb{R}$,
and the functional $g$ is H\"{o}lder continuous.  That is, there exists $K>0$ and $\alpha > 0$ such that
$$
  |g(w)-g(w')| \le K \|w-w'\|_U^\alpha \quad\forall\, w,\,w'\in U.
$$
We use the local RB method to build a surrogate model $\overline{u}(\xi):\Xi\rightarrow U$ and evaluate $\overline{l}(\xi)$ as $\overline{l}(\xi) = g(\overline{u}(\xi))$.
Note that based on the above assumption, we have that
\begin{equation}\label{eq:loss_holder_bdd}
|l(\xi)-\overline{l}(\xi)| \le K \|u(\xi) - \overline{u}(\xi)\|_U^{\alpha}\quad\forall \xi\in\Xi. 
\end{equation}

In the local RB method, we partition the parameter space into Voronoi cells, i.e., 
$\Xi = \cup_{k=1}^n \Xi_k$, seeded at $n$ selected atoms $\xi_k$, $k=1,\dots,n$.  
Within each cell, we form a local basis for approximating the PDE solution $u(\xi)$ using, e.g., full-order PDE solutions at a fixed number of proximal atoms as well as the gradient of the solution at the given seed.  
For example, Figure~\ref{fig_localRB} shows a partition of a parameter domain $\Xi\subseteq\mathbb{R}^2$ with 2,000 Monte Carlo samples of $\xi$ in the background (e.g., drawn from a prior distribution).  
The surrogate solution at the large blue dot is computed using a basis consisting of full PDE solutions at the large solid red dots as well as the solution and gradient at the large black dot.  
The number of neighbors $N$ for the local basis is usually chosen to be a fixed constant. 
In general, the number of neighbors is an algorithmic choice of the user, but can also be chosen adaptively depending on the desired accuracy of the local approximation.
\begin{figure}[!ht]
\centering
\includegraphics[width=0.4\linewidth]{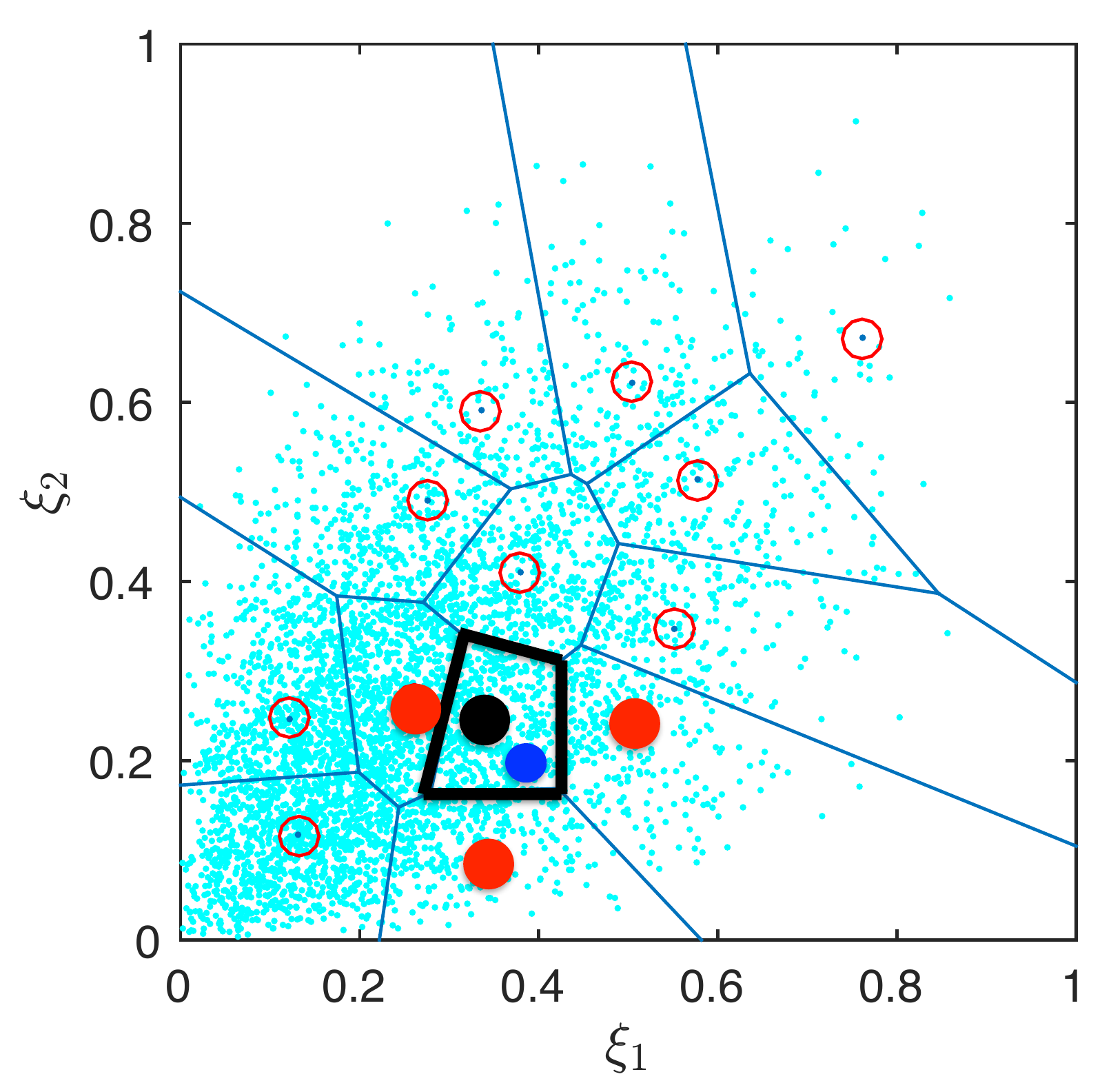}
\caption{ The local reduced basis method with two random parameters.
The surrogate solution at the large blue dot is computed using a basis
consisting of the solution at the large solid red dots as well as the solution
and gradient at the large black dot.  We plot 2,000 Monte Carlo samples of
$\xi$ in the background.}
\label{fig_localRB}
\end{figure}

The local RB surrogate model $\overline{u}(\xi)$ of $u(\xi)$ is given by
\begin{equation} \label{LRB_Surro}
  \overline{u}(\xi) = \sum_{k=1}^{n}\mathbbm{1}_{\Xi_k}(\xi)u_k(\xi)
\end{equation}
where $\mathbbm{1}_{\Xi_k}$ denotes the characteristic function of the
set $\Xi_k$, i.e., $\mathbbm{1}_{\Xi_k}(\xi) = 1$ if $\xi\in\Xi_k$ and
$\mathbbm{1}_{\Xi_k}(\xi) = 0$ otherwise, and $u_k:\Xi_k\to U_k$
is the solution of the reduced problem
\[
  \langle M(u_k(\xi); \xi), v\rangle_{V^*,V} = 0
    \quad\forall\,v\in V_k, \quad\forall\,\xi\in\Xi_k.
\]
Here, $\Phi_k$ is a ``local basis'' within $\Xi_k$, e.g.,
$$
 \Phi_k = \left[ u(\xi_k), \nabla_{\xi}u(\xi_k),  u(\xi_{k_1}),
      u(\xi_{k_2}),\dots, u(\xi_{k_N}) \right],
$$
$U_k= \textup{span}(\Phi_k)$, and $V_k$ is a finite-dimensional subspace of
$V$. Since the cardinality of $\Phi_k$ is typically much smaller than the full
discretization of the PDE, we often realize significant computational savings
by using $\overline{u}(\xi)$ as a surrogate model for $u(\xi)$. 
In addition, due to the local nature of the approximation, the evaluation cost
of $\overline{u}(\xi)$ at any $\xi\in\Xi$ does not increase as the number
of atoms $n$ increases.

To efficiently construct the local RB surrogate, we employ a greedy adaptive
sampling procedure to select the atom set $\Theta:=\{\xi_k\}_{k=1}^n$
(for details on this algorithm see \cite{zou2018adaptive}).  The adaptive selection of $\Theta$ is guided by reliable
a posteriori error indicators, denoted by $\epsilon_u(\xi)$, i.e.,
\[
  \|\overline{u}(\xi) - u(\xi)\|_U \lesssim \epsilon_u(\xi)
\]
where $x \lesssim y$ denotes ``$x$ is less than or equal to a constant times
$y$.'' That is, given $k$ atoms, the next atom $\xi_{k+1}$ is selected from the
region of $\Xi$ where the current surrogate error is the largest.  The error
indicators $\epsilon_u(\xi)$ used in \cite{zou2018adaptive} are residual-based error
estimates. In fact, we have shown in \cite{zou2018adaptive} that the error indicator
$\epsilon_u(\xi)$ can be further used to build more complex error indicators
that are specifically targeted for the approximating quantities of interest such as risk measures. 

For example, a possible error indicator for $\overline{l}(\xi)$ can be derived using Equation \eqref{eq:loss_holder_bdd} as
\begin{equation}\label{eq:loss_holder_bdd_error_indicator}
|l(\xi)-\overline{l}(\xi)|  \lesssim \epsilon_l(\xi) := K \epsilon_u(\xi)^{\alpha}  \quad\forall \xi\in\Xi,
\end{equation}
which we can use to tailor the local RB method to specifically approximate $l(\xi)$. 


\section{Particle based approximation with local RB surrogate}
A particle based approximation depends on a set of weighted samples $\{\xi_i, w_i\}_{i=1}^m$, which defines the following empirical measure
\begin{equation}\label{eq:gibbsEmpDist}
\rho^E(\xi) = \sum_{i=1}^m w_i \delta(\xi-\xi_i)
\end{equation}
where $\delta(\cdot)$ is the Dirac delta function and the weights $w_i$ satisfy $w_i\in[0, 1]$ and $\sum_{i=1}^m w_i=1$. 
This form of approximation has been used extensively in sequential Monte Carlo (SMC) methods \cite{del2006sequential, doucet2001introduction}. 
Here, we use such an approximation for the Gibbs posterior.

To approximate a given distribution $\rho(\xi)$, $\xi_i$ and $w_i$ need to be selected in some principled manner, e.g., by minimizing some distance between $\rho^E(\xi)$ and $\rho(\xi)$. For instance,
a set of Monte Carlo (MC) samples drawn from $\rho(\xi)$ with equal weights is a particle approximation to $\rho(\xi)$. 
However, MC approximations typically display slow convergence and high variance. We mention in passing that there are more efficient methods for constructing particle approximations such as the Stochastic Reduced Order Model approach introduced in  \cite{grigoriu2009reduced, grigoriu2012method}. 


Now, suppose we have a particle based approximation for the prior $\rho_0(\xi)$ based on the particle set $\{\xi_i, w_i^0\}_{i=1}^m$.
Using the empirical measure $\rho^E_0(\xi) := \sum_{i=1}^m w_i^0 \delta(\xi-\xi_i)$ in \eqref{eq:gibbsPosFormula_L}, we have an update formula for the particle weights as
\begin{equation}\label{eq:gibbsPosFormulaWts}
w_i = \frac{\exp(-Wl(\xi_i))w_i^0}{\sum_{k=1}^{m} w_k^0 \exp(-Wl(\xi_k))}.
\end{equation}
where $w_i$ is the posterior weight associated with $\xi_i$. 
By \eqref{eq:gibbsEmpDist}, $\{\xi_i, w_i\}_{i=1}^m$ defines a particle based approximation $\rho^E(\xi)$ to the Gibbs posterior distribution $\rho(\xi)$ in \eqref{eq:gibbsPosFormula_L}. 

If we use a surrogate model $\overline{l}(\xi)$ for the loss function, the posterior weights are then approximated as
\begin{equation}\label{eq:gibbsPosFormulaWtsRB}
\overline{w}_i = \frac{\exp(-W\overline{l}(\xi_i))w_i^0}{\sum_{k=1}^{m} w_k^0 \exp(-W\overline{l}(\xi_k))}.
\end{equation}
The particles $\{\xi_i, \overline{w}_i\}_{i=1}^m$ define a surrogate empirical posterior measure
\begin{equation}\label{eq:gibbsEmpDistRB}
\overline{\rho}^E(\xi) = \sum_{i=1}^m \overline{w}_i \delta(\xi-\xi_i)
\end{equation}
that is an approximation to $\rho^E(\xi)$.

\subsection{Accuracy of surrogate particle approximations}
We now undertake an error analysis of the above approximations based on a fixed set of particles. 
The evolution of particles will be addressed in the following section using the framework of SMC methods. 
The results shown here are cornerstones for the convergence of the SMC method described in a subsequent section.  
We first state a lemma regarding a randomized particle based approximation.
\begin{lemma}\label{lm:gibbsMCBdd}
Given a distribution $\rho_0(\xi)$, let $\{\xi_i\}_{i=1}^m$ be drawn independently from $\rho_0(\xi)$, $w_i^0$ be the associated probability weight, and $\rho_0^E(\xi)$ be the empirical distribution in \eqref{eq:gibbsEmpDist}, then we have that
\begin{align}
h(\rho_0^E, \rho_0)\le \sqrt{\sum_{i=1}^m (w_i^0)^2}.
\end{align}
In particular, if $w_i^0 = \frac{1}{m}$, as for the MC weights, we have that
\begin{align}
h(\rho_0^E, \rho_0)\le \frac{1}{\sqrt{m}} .
\end{align}
\end{lemma}
\begin{proof}
For any $f$,  we have 
$$\rho_0^E[f] = \sum_{i=1}^m w_i^0  f(\xi_i).$$
Hence,
$$\rho_0^E[f]- \rho_0[f] = \sum_{i=1}^m w_i^0  (f(\xi_i)-\rho_0[f]).$$
Note that since the particles $\xi_i$ are i.i.d. with distribution $\rho_0(\xi)$, we have $\mathbb{E}[f(\xi_i)] = \rho_0[f]\ \forall\ i=1,\dots, m$. 
Hence
$$\mathbb{E}[(f(\xi_i)-\rho_0[f])(f(\xi_j)-\rho_0[f])]=\delta_{ij}\mathbb{E}[|(f(\xi_i)-\rho_0[f])|^2].$$
Additionally, since $|f|_{\infty}\le 1$, we have
$$\mathbb{E}[|(f(\xi_i)-\rho_0[f])|^2] = \mathbb{E}[|f(\xi_i)|^2]-\rho_0[f]^2\le 1.$$
Therefore,
$$\mathbb{E}[|\rho_0^E[f]-\rho_0[f]|^2]=\sum_{i=1}^m (w_i^0)^2 \mathbb{E}[|(f(\xi_i)-\rho_0[f])|^2] \le \sum_{i=1}^m (w_i^0)^2,$$
which completes the proof. 
\end{proof}

\begin{lemma}\label{lm:gibbsMapBdd}
Denote the transformation of the Gibbs posterior formula in \eqref{eq:gibbsPosFormula_L} as $G_W:P\rightarrow P$, such that $\rho(\xi) = G_W\rho_0(\xi)$.
Under Assumption \ref{as:gibbsLoss}, we have 
\begin{align}
h(G_W \rho_1, G_W \rho_2) \le 2\exp(WC_l) h(\rho_1, \rho_2),
\end{align}
for $\forall\ \rho_1, \rho_2\in P$.
\end{lemma}
\begin{proof}
We have 
\begin{align}
G_W \rho_1[f]-G_W \rho_2[f] &=  \frac{\int_{\Xi}\exp(-Wl(\xi))\rho_{1}(\xi)f(\xi)d\xi}{\int_{\Xi} \exp(-Wl(\xi))\rho_{1}(\xi)d\xi}-\frac{\int_{\Xi}\exp(-Wl(\xi))\rho_{2}(\xi)f(\xi)d\xi}{\int_{\Xi} \exp(-Wl(\xi))\rho_{2}(\xi)d\xi}\nonumber\\
&=\frac{\int_{\Xi} \exp(-Wl(\xi))(\rho_{1}(\xi)-\rho_{2}(\xi))f(\xi)d\xi}{Z_3}\nonumber\\
&+\frac{(Z_4-Z_3)\int_{\Xi} \exp(-Wl(\xi))\rho_{2}(\xi)f(\xi)d\xi}{Z_3Z_4}\nonumber
\end{align}
where
\begin{align} \label{eq:Z3Z4}
Z_3= \int_{\Xi} \exp(-Wl(\xi))\rho_{1}(\xi)d\xi,\quad Z_4= \int_{\Xi} \exp(-Wl(\xi))\rho_{2}(\xi)d\xi.
\end{align}
Hence
\begin{align}
|G_W \rho_1[f]-G_W \rho_2[f]| &\le \frac{|\int_{\Xi} \exp(-Wl(\xi))(\rho_{1}(\xi)-\rho_{2}(\xi))f(\xi)d\xi|}{Z_3}\nonumber\\
&+\frac{|Z_4-Z_3||\int_{\Xi} \exp(-Wl(\xi))\rho_{2}(\xi)f(\xi)d\xi|}{Z_3Z_4}.\nonumber
\end{align}
Note that since $|f|_{\infty}\le 1$, 
$$\left|\int_{\Xi} \exp(-Wl(\xi))\rho_{2}(\xi)f(\xi)d\xi\right|\le \int_{\Xi} \exp(-Wl(\xi))\rho_{2}(\xi)d\xi = Z_4$$
Using Assumption \ref{as:gibbsLoss}, we have
$$Z_3 =  \int_{\Xi} \exp(-Wl(\xi))\rho_{1}(\xi)d\xi \ge \exp(-W C_l)$$
hence we have
\begin{align}
|G_W \rho_1[f]-G_W \rho_2[f]| &\le \frac{|\int_{\Xi} \exp(-Wl(\xi))(\rho_{1}(\xi)-\rho_{2}(\xi))f(\xi)d\xi|}{\exp(-W C_l)}+\frac{|Z_4-Z_3|}{\exp(-W C_l)}.\nonumber
\end{align}
Note that $|\exp(-Wl(\xi))|_{\infty}\le 1$ and $|\exp(-Wl(\xi)) f(\xi)|_{\infty}\le 1$ for $\forall f$ such that $|f|_{\infty}\le 1$, hence
$$
\left|\int_{\Xi} \exp(-Wl(\xi))(\rho_{1}(\xi)-\rho_{2}(\xi))f(\xi)d\xi\right| \le \sup_{|g|_{\infty}\le 1} |\rho_1[g]-\rho_2[g]|
$$
and
$$
|Z_4-Z_3| = \left|\int_{\Xi}(\rho_{1}(\xi)-\rho_{2}(\xi))\exp(-Wl(\xi))d\xi\right| \le \sup_{|g|_{\infty}\le 1} |\rho_1[g]-\rho_2[g]|.
$$
Therefore,
$$|G_W \rho_1[f]-G_W \rho_2[f]| \le 2\exp(WC_l) \sup_{|g|_{\infty}\le 1} |\rho_1[g]-\rho_2[g]|$$
for $\forall f$ such that $|f|_{\infty}\le 1$. The lemma follows easily from the above inequality.
\end{proof}

Next we consider the surrogate approximation. We first state a counterpart of Lemma \ref{lm:gibbsMapBdd} when the surrogate model \eqref{eq:gibbsPosFormula_L_RB} is used instead of \eqref{eq:gibbsPosFormula_L}:
\begin{lemma}\label{lm:gibbsMapBdd_RB}
Denote the transformation of Gibbs posterior formula in \eqref{eq:gibbsPosFormula_L_RB} as $\overline{G}_W:P\rightarrow P$, such that $\overline{\rho}(\xi) = \overline{G}_W \rho_0(\xi)$.
Under Assumption \ref{as:gibbsLoss}, we have 
\begin{align}
h(\overline{G}_W  \rho_1, \overline{G}_W \rho_2) \le 2\exp(W C_{\overline{l}}) h(\rho_1, \rho_2),
\end{align}
for $\forall\ \rho_1, \rho_2\in P$.
\end{lemma}
\begin{proof}
Same as Lemma \ref{lm:gibbsMapBdd}. 
\end{proof}

We impose the following additional assumption on the local RB surrogate loss function $\overline{l}(\xi)$.
\begin{assumption}\label{as:RBLoss}
We assume that the surrogate loss function $\overline{l}(\xi)$ is accurate over the particle set, that is, 
\begin{equation}\label{eq:gibbs_rb_error_bound}
\sup_{i=1,\dots, m} |l(\xi_i) - \overline{l}(\xi_i)|\le e
\end{equation}
for some $e > 0$ that indicates the error of the local RB approximation.
\end{assumption}
The bound \eqref{eq:gibbs_rb_error_bound} can be satisfied for any $e$ by using the local RB Algorithm in \cite{zou2018adaptive} using $\{\xi_i\}_{i=1}^m$ as training samples and
$\epsilon_l(\xi)$ in Equation \eqref{eq:loss_holder_bdd_error_indicator} as the error indicator. 
We have the following two lemmas quantifying the difference between $\rho^E$ with weights computed by \eqref{eq:gibbsPosFormulaWts} and $\overline{\rho}^E$ in \eqref{eq:gibbsEmpDistRB}.
\begin{lemma}\label{lm:gibbsRBKLBdd}
Under Assumptions \ref{as:gibbsLoss} and \ref{as:RBLoss}, the following bound holds:
\begin{align}
    D_{KL}(\overline{\rho}^E|\rho^E) \le 2 We.\nonumber
  \end{align}
\end{lemma}
\begin{proof}
  Based on Assumption \ref{as:RBLoss}, we have
   \begin{align}
   \frac{\exp(-W\overline{l}(\xi_i))}{\exp(-Wl(\xi_i))} = \exp(Wl(\xi_i)-W\overline{l}(\xi_i))\le \exp(We)\ \ \forall\ i=1,\dots, m\nonumber\\
   \frac{\exp(-Wl(\xi_i))}{\exp(-W\overline{l}(\xi_i))} = \exp(W\overline{l}(\xi_i)-Wl(\xi_i))\le \exp(We)\ \ \forall\ i=1,\dots, m.\nonumber\nonumber
   \end{align}
  We first define
   \begin{align}\label{eq:Z5Z6}
   Z_5 = \sum_{k=1}^{m} w_k^0 \exp(-Wl(\xi_k)), \quad Z_6 = \sum_{k=1}^{m} w_k^0 \exp(-W\overline{l}(\xi_k)).
   \end{align}
   It is clear that
   $$\log\left(\frac{Z_5}{Z_6}\right)\le \log\left(\frac{\exp(We) \sum_{k=1}^{m} w_k^0 \exp(-W\overline{l}(\xi_k))}{\sum_{k=1}^{m}w_k^0 \exp(-W\overline{l}(\xi_k))}\right)=W e.$$
  Hence
  \begin{align}
    D_{KL}(\overline{\rho}^E|\rho^E) &= \sum_{k=1}^{m} \frac{\exp(-W\overline{l}(\xi_k))w_k^0}{Z_6}\log\left(\frac{\exp(-W\overline{l}(\xi_k))w_k^0 Z_5}{\exp(-Wl(\xi_k))w_k^0 Z_6}\right)\nonumber\\
    &\le \sum_{k=1}^{m} \frac{\exp(-W\overline{l}(\xi_k))w_k^0}{Z_6} \left[\log\left(\frac{Z_5}{Z_6}\right)+W e\right]\nonumber\\
    & = \log\left(\frac{Z_5}{Z_6}\right)+W e = 2 W e.\nonumber    
   \end{align}
\end{proof}

\begin{lemma}\label{lm:gibbsRBmetricBdd}
Under Assumptions \ref{as:gibbsLoss} and \ref{as:RBLoss}, the following bound holds:
\begin{align}
    h(\rho^E, \overline{\rho}^E) \le 2\exp(W C_l)C W e ,\nonumber
  \end{align}
  for some constant $C > 0$. 
\end{lemma}
\begin{proof}
Recall the definition of $Z_5$ and $Z_6$ in \eqref{eq:Z5Z6}, we have
\begin{align}
\rho^E[f]-\overline{\rho}^E[f] & =  \frac{\sum_{j=1}^m\exp(-Wl(\xi_j))w_j^0 f(\xi_j)}{\sum_{k=1}^{m} \exp(-Wl(\xi_k))w_k^0}  - \frac{\sum_{j=1}^m\exp(-W\overline{l}(\xi_j))w_j^0 f(\xi_j)}{\sum_{k=1}^{m} \exp(-W\overline{l}(\xi_k))w_k^0} \nonumber\\
&=\frac{\sum_{j=1}^m(\exp(-Wl(\xi_j))-\exp(-W\overline{l}(\xi_j)))w_j^0 f(\xi_j)}{Z_5} \nonumber\\
&+\frac{(Z_6-Z_5)\sum_{j=1}^m\exp(-W\overline{l}(\xi_j))w^{0}_j f(\xi_j)}{Z_5Z_6}.\nonumber
\end{align}
Note that
$$Z_5 = \sum_{k=1}^{m} \exp(-Wl(\xi_k))w^{0}_k\ge \exp(-W C_l),$$
and that
$$\left|\sum_{j=1}^m\exp(-W\overline{l}(\xi_j))w^{0}_j f(\xi_j)\right|\le Z_6,$$
since $|f|_{\infty}\le 1$. Also, since $-e\le l(\xi_j)-\overline{l}(\xi_j)\le e$, for $e$ sufficiently small, e.g., $We\ll 1$, we have
\begin{align}
|\exp(-Wl(\xi_j))-\exp(-W\overline{l}(\xi_j))| &\le \exp(-Wl(\xi_j)) |1-\exp(Wl(\xi_j)-W\overline{l}(\xi_j))| \le C W e\nonumber
\end{align}
for some $C > 0$. Hence
$$|Z_5-Z_6|\le C W e.$$
Finally, we have
\begin{align}
|\rho^E[f]-\overline{\rho}^E[f]| &\le 2\exp(W C_l)C W e, \nonumber
\end{align}
which implies the Lemma.
\end{proof}

We now show one of the main analytical results of our work. Namely, we provide an error bound for a particle-based approximation to the Gibbs posterior.

\begin{theorem}\label{thm:gibbsBdd}
For an empirical distribution $\rho^E_0$ based on $\{\xi_i, w^0_i\}_{i=1}^m$, which is an approximation to the prior $\rho_0$, the approximate posterior $\overline{\rho}^E$ defined by \eqref{eq:gibbsEmpDistRB} satisfies
\begin{align}
    h(\overline{\rho}^E, \rho) \le 2\exp(W C_l)C W e + 2\exp(WC_l) \sqrt{\sum_{i=1}^m (w_i^0)^2}
 \end{align}
for the same constants $C > 0$ as in Lemma \ref{lm:gibbsRBmetricBdd}.  In particular, if $w^0_i = \frac{1}{m}$, i.e., the particles are MC samples of the prior, we have
$$h(\overline{\rho}^E, \rho) \le 2\exp(W C_l)C W e + 2\exp(WC_l) \frac{1}{\sqrt{m}}.$$
 \end{theorem}
 \begin{proof}
 By triangle inequality and Lemma \ref{lm:gibbsMCBdd}, \ref{lm:gibbsMapBdd} and \ref{lm:gibbsRBmetricBdd}, we have that
 \begin{align}
h(\overline{\rho}^E, \rho) &\le h(\overline{\rho}^E, \rho^E) + h(\rho^E, \rho)\nonumber\\
&\le h(\overline{\rho}^E, \rho^E) + 2\exp(WC_l) h(\rho^E_0, \rho_0)\nonumber\\
&\le  2\exp(W C_l)C W e  + 2\exp(WC_l) \sqrt{\sum_{i=1}^m (w_i^0)^2}.\nonumber
\end{align}
This completes the proof.
\end{proof}
From these results, we can see that if we use an MC approximation to the prior, we can make $h(\overline{\rho}^E, \rho)$ arbitrarily small by decreasing $e$ and increasing $m$. 
In practice, however, it is nontrivial to do both at the same time, as when the number of particles $m$ increases, 
we require stronger global accuracy on our surrogate model $\overline{l}(\xi)$,  which can only be achieved by globally refining the surrogate with more PDE solves.
To efficiently represent the posterior distribution with a limited number of particles, we rely on a SMC method to progressively evolve the particles through a sequence of interpolating distributions from $\rho_0$ to $\rho$.
The samples evolved through the local RB surrogate are automatically navigated to the support of the posterior in the process. 

\section{An adaptive Sequential Monte Carlo method for particle evolution}
The above discussion deals with the asymptotic convergence of the particle approximation based on a fixed set of particles. 
In practice, using \eqref{eq:gibbsPosFormulaWtsRB} as an approximation to the posterior weights with a fixed set of particles drawn from $\rho_0$ may lead to a poor approximation to $\rho$, especially when $W$ is large. 
The reason for this is the potential loss of sample diversity, i.e., the posterior mass may concentrate over just a few particles.   

In SMC, instead of computing the posterior weights at once, particles are evolved to approximate a sequence of intermediate distributions interpolating from the prior to posterior. 
The particles are resampled and mutated after each iteration to prevent degeneracy. We adopt such an SMC framework to approximate the Gibbs posterior distribution. 
In particular, our SMC method for Gibbs posterior follows closely to the recent work in \cite{kantas2014sequential} where the authors proposed an SMC method for high dimensional inverse problems. 

\subsection{The Sequential Monte Carlo method}
In the context of Gibbs posterior, the sequence of the interpolating distributions are defined by 
\begin{equation}\label{eq:intermediate_dist_1}
\rho_t = G_{W_t}\rho_0,\ 0\le t\le N
\end{equation}
where $0=W_0< W_1< W_2\dots<W_t<\dots<W_N = W$, and recall the definition of $G_{W}$ as the Gibbs update formula defined in \eqref{eq:gibbsPosFormula_L}. we set $\rho_N = \rho$, which is the posterior distribution we want to approximate. 
Also, it is easy to show that we have the following property
\begin{equation}\label{eq:intermediate_dist_2}
\rho_t = G_{W_t-W_s}\rho_s,\ 0\le s\le t\le N
\end{equation}
by the Gibbs update formula \eqref{eq:gibbsPosFormula_L}. 
This property allows us to apply SMC methods and progressively approximate $\rho$. 

As mentioned, the key idea of SMC is to start from a particle based approximation of $\rho_0$, i.e., $\rho_0^E$, which is easy to obtain, and gradually increase the weight $W_t$ until it reaches $W$, adjusting the particles along the way. 
To this end, we denote the particle approximation to $\rho_t$ as $\rho_t^E$,
$$\rho_t^E = \sum_{i=1}^m w^t_i \delta(\xi-\xi^t_i)$$
based on the particle set $\{\xi^t_i, w^t_i\}_{i=1}^m$. 

The iteration $t+1$ of the SMC involves three steps: 
(i) update the weights of the current particle set $\{\xi^t_i\}_{i=1}^m$ by
\begin{equation}\label{eq:gibbsPosFormulaWts_smc}
w^{t+1}_i = \frac{\exp(-(W_{t+1}-W_t)l(\xi^t_i))w^{t}_i }{\sum_{k=1}^{m} w^{t}_k \exp(-(W_{t+1}-W_t)l(\xi^t_k))}.
\end{equation}
The distribution based on $\{\xi^t_i, w^{t+1}_i\}_{i=1}^m$ is denoted by 
$$\rho_{t+1, t}^E = \sum_{i=1}^m w^{t+1}_i \delta(\xi-\xi^t_i),$$
that is $\rho_{t+1, t}^E  = G_{W_{t+1}-W_t} \rho_{t}^E$. 
By Lemma \ref{lm:gibbsMapBdd}, we have that
\begin{equation}\label{eq:SMC_step1_bdd}
h(\rho_{t+1, t}^E, \rho_{t+1})\le 2\exp((W_{t+1}-W_t) C_l) h(\rho_{t}^E, \rho_{t}).
\end{equation}

(ii) Resample the particles $\{\xi^t_i\}_{i=1}^m$ with replacement according to the weights $\{w^{t+1}_i\}_{i=1}^m$, 
i.e., resample according to $\rho_{t+1, t}^E$. 
This step effectively eliminates the particles with negligible weights and duplicate the particles with large weights. 
All resampled particles, including the duplicates, are denoted by $\{\xi^{t+1,t}_i\}_{i=1}^m$ and are assigned equal weights $1/m$.  
The resampled distribution is denoted by 
$$\rho^{E,S}_{t+1, t} = \sum_{i=1}^m \frac{1}{m} \delta(\xi-\xi^{t+1,t}_i) .$$
By Lemma \ref{lm:gibbsMCBdd}, we have that
\begin{equation}\label{eq:SMC_step2_bdd}
h(\rho^{E,S}_{t+1, t}, \rho^{E}_{t+1, t})\le \frac{1}{\sqrt{m}}.
\end{equation}
This resampling step alone does not prevent sample degeneracy, as only a few  particles will survive and copy themselves. 
To preserve population diversity, a third step is required. 

(iii) Apply a $\rho_{t+1}$-invariant mutation to the resampled set $\{\xi^{t+1,t}_i\}_{i=1}^m$ in step (ii). 
This can be achieved by evolving the particles $\{\xi^{t+1,t}_i\}_{i=1}^m$ independently by one or more steps using a $\rho_{t+1}$-invariant Markov kernel $\mathcal{K}_{t+1}$ (i.e., $\rho_{t+1}=\rho_{t+1}\mathcal{K}_{t+1}$), 
e.g., an MCMC kernel with $\rho_{t+1}$ as the stationary distribution. Note that the invariant property of $\mathcal{K}_{t+1}$ implies that \cite{ beskos2015sequential, rebeschini2015can},
\begin{equation}\label{eq:markov}
h(p\mathcal{K}_{t+1}, q\mathcal{K}_{t+1}) \le h(p, q),\ \forall\ p, q\in P,\ \forall\ 0\le t\le N-1.
\end{equation}
The resulted particles from step (iii), denoted by $\{\xi^{t+1}_i\}_{i=1}^m$, with weights $1/m$, 
define the distribution 
$$\rho^E_{t+1} = \sum_{i=1}^m \frac{1}{m} \delta(\xi-\xi^{t+1}_i)$$ 
that is used to approximate $\rho_{t+1}$ and is used for the next iteration of the SMC. 
We adopt the same MCMC mutation kernel proposed in \cite{kantas2014sequential} for this step, which has been shown to be efficient for high dimensional inverse problems. 

We have the following theorem regarding the SMC method for the Gibbs posterior.
\begin{theorem}\label{lm:gibbsBdd_SMC}
Assuming that the initial particles are a set of MC samples with equal weights $1/m$, then following the above outlined SMC method with the exact loss function $l(\xi)$, we have that for all iterations $t: 0\le t\le N-1$,
\begin{equation}\label{eq:gibbsBdd_SMC_Step}
h(\rho^E_{t+1}, \rho_{t+1})\le \frac{1}{\sqrt{m}} \sum_{s=0}^{t+1} 6^{t+1-s}\exp((W_{t+1}-W_s) C_l)
\end{equation}
in particular, we have a posterior error bound
\begin{equation}\label{eq:gibbsBdd_SMC_Pos}
h(\rho^E, \rho)\le \frac{1}{\sqrt{m}} \sum_{s=0}^{N} 6^{N-s}\exp((W-W_s) C_l).
\end{equation}
where $W=W_N$.
\end{theorem}
\begin{proof}
First, by Equation \eqref{eq:SMC_step2_bdd}, \eqref{eq:markov} and the fact that $\rho_{t+1}=\rho_{t+1}\mathcal{K}_{t+1}$, we have
\begin{align}
h(\rho^E_{t+1}, \rho^{E}_{t+1, t}) &= h(\rho^{E,S}_{t+1, t}\mathcal{K}_{t+1}, \rho^{E}_{t+1, t})\le h(\rho^{E,S}_{t+1, t}\mathcal{K}_{t+1}, \rho^{E}_{t+1, t}\mathcal{K}_{t+1}) + h(\rho^{E}_{t+1, t}\mathcal{K}_{t+1}, \rho^{E}_{t+1, t})\nonumber\\
&\le h(\rho^{E,S}_{t+1, t}, \rho^{E}_{t+1, t}) + h(\rho^{E}_{t+1, t}\mathcal{K}_{t+1}, \rho_{t+1} \mathcal{K}_{t+1}) + h(\rho_{t+1}, \rho^{E}_{t+1, t})\nonumber\\
&\le  \frac{1}{\sqrt{m}} + 2h(\rho_{t+1}, \rho^{E}_{t+1, t}).\nonumber
\end{align}
Hence
\begin{align}
h(\rho^E_{t+1}, \rho_{t+1})&\le h(\rho^E_{t+1}, \rho^{E}_{t+1, t}) + h(\rho_{t+1}, \rho^{E}_{t+1, t}) \le \frac{1}{\sqrt{m}} + 3 h(\rho_{t+1}, \rho^{E}_{t+1, t})\nonumber\\
&\le \frac{1}{\sqrt{m}} + 6\exp((W_{t+1}-W_t) C_l) h(\rho_{t}^E, \rho_{t})\nonumber
\end{align}
by Equation \eqref{eq:SMC_step1_bdd}. Iterating gives
$$h(\rho^E_{t+1}, \rho_{t+1})\le \frac{1}{\sqrt{m}} \sum_{s=0}^{t+1} 6^{t+1-s}\exp((W_{t+1}-W_s) C_l),$$
which completes the proof.
\end{proof}

When a local RB surrogate loss function $\overline{l}(\xi)$ is used, the sequence of distributions are defined by $\overline{\rho_t^E}$. The update in step (i) is replaced by
\begin{equation}\label{eq:gibbsPosFormulaWts_smc_RB}
\overline{w^{t+1}_i} = \frac{\exp(-(W_{t+1}-W_t)\overline{l}(\xi^t_i))\overline{w^{t}_i} }{\sum_{k=1}^{m} \overline{w^{t}_k} \exp(-(W_{t+1}-W_t)\overline{l}(\xi^t_k))}.
\end{equation}
which defines $\overline{\rho_{t+1, t}^E} = \sum_{i=1}^m \overline{w^{t+1}_i} \delta(\xi-\xi^t_i)$. 
That is, $\overline{\rho_{t+1, t}^E} = \overline{G}_{W_{t+1}-W_t} \overline{\rho_{t}^E} $. By Lemma \ref{lm:gibbsMapBdd_RB}, we have
\begin{equation}\label{eq:SMC_step1_bdd_RB}
h(\overline{\rho_{t+1, t}^E} , \overline{\rho_{t+1}})\le 2\exp((W_{t+1}-W_t) C_{\overline{l}}) h(\overline{\rho_{t}^E}, \overline{\rho_{t}}), 
\end{equation}
where
\begin{equation}\label{eq:gibbsPosFormula_L_RB_step}
\overline{\rho_{t}}(\xi)= \frac{\exp(-W_{t}\overline{l}(\xi))\rho_{0}(\xi)}{\int_{\Xi} \exp(-W_{t}\overline{l}((\xi))\rho_{0}(\xi)d\xi}.
\end{equation}

In addition, the kernel mutation step requires evaluation of the loss function $l(\xi)$ at new parameters which can be accelerated by $\overline{l}(\xi)$ as well. 
To this end, we use a surrogate kernel $\overline{\mathcal{K}_{t+1}}$ that is invariant with respect to $\overline{\rho_{t+1}}$. 
The mutation with respect to $\overline{\mathcal{K}_{t+1}}$ only requires evaluation of $\overline{l}(\xi)$.
We have the following theorem regarding the SMC method using $\overline{l}(\xi)$:
\begin{theorem}\label{lm:gibbsBdd_SMC_RB}
Assuming that the initial particles are a set of MC samples with equal weights $1/m$, 
then following the above outlined SMC method using the local RB surrogate $\overline{l}(\xi)$ for step (i) and (iii), 
we have that for all iterations $t: 0\le t\le N-1$,
\begin{align}\label{eq:gibbsBdd_SMC_Step_RB}
h(\overline{\rho^E_{t+1}}, \rho_{t+1}) &\le \frac{1}{\sqrt{m}} \sum_{s=0}^{t+1} 6^{t+1-s}\exp((W_{t+1}-W_s) C_{\overline{l}}) + 2\exp(W_{t+1}C_l) C_e W_{t+1} e \nonumber\\
&+ 2\exp(W_{t+1}C_l + W_{t+1}\max\{C_l, C_{\overline{l}}\}) \rho_{t+1}[\mathbbm{1}_{\Xi_e^{\perp}}]
\end{align}
In particular, we have the posterior error bound
\begin{align}\label{eq:gibbsBdd_SMC_Pos_RB}
h(\overline{\rho}^E, \rho) &\le  \frac{1}{\sqrt{m}} \sum_{s=0}^{N} 6^{N-s}\exp((W-W_s) C_{\overline{l}}) + 2\exp(WC_l) C_e W e \nonumber\\
&+ 2\exp(WC_l + W\max\{C_l, C_{\overline{l}}\}) \rho[\mathbbm{1}_{\Xi_e^{\perp}}],
\end{align}
where $W=W_N$.
\end{theorem}
\begin{proof}
The first term on the right-hand-side comes from a simple restatement of Theorem \ref{lm:gibbsBdd_SMC} for $h(\overline{\rho^E_{t+1}}, \overline{\rho_{t+1}})$. 
The remainders of the right-hand-side is due to the error bound in Theorem \ref{thm:gibbsBdd_continuous}.
\end{proof}

From  Theorem \ref{lm:gibbsBdd_SMC_RB}, we see how we should construct the surrogate model $\overline{l}(\xi)$. 
Given the number of particles $m$ and the prescribed surrogate accuracy $e$, we need to build the surrogate $\overline{l}(\xi)$ so that $\rho[\mathbbm{1}_{\Xi_e^{\perp}}]$, 
i.e. the posterior measure of the ``unfeasible set"  $\Xi_e^{\perp}$, is minimized. 
In terms of local RB surrogate, this requires concentration of local RB atoms and the accurate approximation of $l(\xi)$ over the support of the posterior.
To this end, we progressively train the local RB surrogate using the sequence of particles $\{\xi^t_i\}_{i=1}^m$. 
As the particles gradually cluster over the support for the posterior through the SMC iterations, 
the focus of the local RB surrogate is automatically navigated to the support of the posterior as well, resulting in a decrease of the measure of the inaccurate set $\rho[\mathbbm{1}_{\Xi_e^{\perp}}]$. In addition, notice that the support of the posterior typically corresponds to a small and local region of the support of the prior. Hence, once the local RB model becomes sufficiently accurate over that region, 
further evolution of the particles does not require expensive updates of the surrogate model, leading to computational savings. 

Notice that the surrogate loss function $\overline{l}(\xi)$ is changing throughout the SMC iterations. We can recover consistency in \eqref{eq:gibbsPosFormula_L_RB_step} for all $t$ 
by re-running the SMC algorithm from the beginning up to the current $W_t$ using the latest $\overline{l}(\xi)$ before the next SMC iteration. 
This procedure is computationally inexpensive since the surrogate model samples do not need to evolve during the re-run and hence no full PDE solves are required. 
Therefore, we always assume that the update \eqref{eq:gibbsPosFormula_L_RB_step} is consistent for all iterations with the latest surrogate model $\overline{l}(\xi)$. 

We now present the MCMC algorithm for the mutation step using the local RB surrogate $\overline{l}(\xi)$.
To this end, we first define the following mean and variance of $\overline{\rho_{t+1, t}^E}$, which is the particle distribution after SMC step (i) and before resampling step (ii), 
for each parameter dimension $j\in\{1, 2, \dots, M\}$ as 
\begin{align}\label{eq:MCMC_stat}
\overline{m_{t+1, t, j}^E }= \sum_{i=1}^m \overline{ w^{t+1}_{i} } \xi^t_{i, j},\quad \overline{\Sigma_{t+1, t, j}^E} = \sum_{i=1}^m \overline{ w^{t+1}_{i} } (\xi^t_{i, j} - \overline{m_{t+1, t, j}^E})^2.
\end{align}
The above quantities provide estimates of the mean and variance of $\overline{\rho_{t+1}}$ along each individual parameter dimension at SMC iteration $t+1$,
and will be used to facilitate the design a proposal distribution for the MCMC kernel $\overline{\mathcal{K}_{t+1}}$.

Based on the above definition, a proposal $\hat{\xi}^{t+1,t}_{i}$ for a particle $\xi^{t+1,t}_i\in\{\xi^{t+1,t}_i\}_{i=1}^m$ can be obtained by the following mutation
\begin{align}\label{eq:MCMC_proposal}
\hat{\xi}^{t+1,t}_{i, j} = \overline{m_{t+1, t, j}^E} + \gamma (\xi^{t+1,t}_{i, j}-\overline{m_{t+1, t, j}^E}) + \sqrt{1-\gamma^2}\Lambda_{t+1, t, j}, \ \ \ 1\le j\le M
\end{align}
where $\gamma$ is an algorithmic constant and $\Lambda_{t+1, t, j}$ is a random variable with distribution $\mathcal{N}(0, \overline{\Sigma_{t+1, t, j}^E})$. 
Note that the scaling of the proposal distribution is tailored for each individual parameter dimension by the variance estimates $\overline{\Sigma_{t+1, t, j}^E}$ to improve mixing.
In contrast to standard random walk proposals, the above proposal scales to high dimensional problems as shown in \cite{kantas2014sequential}.
The transition probability associated with the proposal in \eqref{eq:MCMC_proposal} is given by
\begin{align}\label{eq:MCMC_transition_pdf}
Q(\hat{\xi}^{t+1,t}_{i} | \xi^{t+1,t}_{i}) = \exp\left(-\frac{1}{2(1-\gamma^2)}\sum_{j=1}^M \frac{ (\hat{\xi}^{t+1,t}_{i, j}-\overline{m_{t+1, t, j}^E} - \gamma(\xi^{t+1,t}_{i, j}-\overline{m_{t+1, t, j}^E}))^2}{\overline{\Sigma_{t+1, t, j}^E}}\right).
\end{align}

Algorithm \ref{al:MCMC_RB} shows the $\overline{\rho_{t+1}}$-invariant mutation MCMC sampler.
\RestyleAlgo{boxruled}
\begin{algorithm}[!ht]
\caption{The MCMC algorithm for $\overline{\rho_{t+1}}$-invariant mutation \label{al:MCMC_RB} }
\begin{minipage}{.922\linewidth}
For each $\xi(0) \in \{\xi^{t+1,t}_i\}_{i=1}^m$, evolve $\xi(0)$ independently for $I$ steps with the following procedure
\begin{itemize}
\setlength\itemsep{0em}
 \item \textbf{For} $i=1,2,\dots, I$, \textbf{do}
  \begin{itemize}\itemsep0em
    \item[--] Draw a proposal $\widehat{\xi(i)}$ using the proposal distribution in \eqref{eq:MCMC_proposal} based on $\xi(i-1)$. 
    \item[--] Use $\overline{l}(\xi)$ to evaluate $\alpha =1\wedge\frac{\exp(-W_{t+1}\overline{l}(\widehat{\xi(i)}))\rho_0(\widehat{\xi(i)})Q(\xi(i-1)|\widehat{\xi(i)})}{\exp(-W_{t+1}\overline{l}(\xi(i-1))\rho_0(\xi(i-1))Q(\widehat{\xi(i)}|\xi(i-1))}$
    \item[--] With probability $1-\alpha$, reject $\widehat{\xi(i)}$ and set $\xi(i) = \xi(i-1)$, $i=i+1$. Go back to the proposal step. 
    \item[--] If $\widehat{\xi(i)}$ not rejected, set $\xi(i) = \widehat{\xi(i)}$ and $i=i+1$. Go back to the proposal step. 
  \end{itemize}
\textbf{End}
\item Finally, return $\xi(I)$ as a $\overline{\rho_{t+1}}$-invariant mutation of $\xi(0)$.
\end{itemize}
\end{minipage}
\end{algorithm}

\subsection{Adaptive selection of the SMC step size}
We now describe how the sequence of step size $0=W_0< W_1< W_2\dots<W_t<\dots<W_N = W$ can be selected adaptively. 
For each SMC iteration $t+1$, we would like to greedily apply all the residual weight $\Delta W = W_{N}-W_t$ to the particle distribution $\overline{\rho_t^E}$ from the previous iteration, 
so that we can directly approximate the posterior $\rho$.
After applying the SMC step (i), i.e., updating the weights by Equation \eqref{eq:gibbsPosFormulaWts_smc_RB} using $\Delta W$, 
we check a simple criteria called the effective sample size (ESS), which is used to measure the sample degeneracy of the current weights 
$$
\text{ESS} = \left(\sum_{k=1}^m \left(\overline{w^{t+1}_i}\right) ^2\right)^{-1}.
$$
Note that $\text{ESS}$ is small if the majority of the probability weights are pivoted on only a few particles, which indicates the loss of sample diversity. 
In this case, we reduce the incremental weights by a constant factor $\Delta W= \theta \Delta W$ with $\theta\in(0, 1)$ 
and repeatedly backtrack and reevaluate Equation \eqref{eq:gibbsPosFormulaWts_smc_RB} and ESS until the latter variable is above some preset threshold. In this case,  
we accept $\Delta W$, set $W_{t+1} = W_t + \Delta W$, move on to the step (ii) and (iii) and finish the current SMC iteration $t+1$. 
If the residual weight is not zero after iteration $t+1$, we set $t=t+1$ and move to the next SMC iteration. Otherwise, we have applied the total weight to the prior and obtained an approximation to the posterior.

Of course, the local RB surrogate model $\overline{l}(\xi)$ evolves as well by the adaptive training on the particles before applying the incremental weight in each SMC iteration. 
We require $\overline{l}(\xi)$ to satisfy Assumption \ref{as:RBLoss} for each iteration $t$. 
As the particles gradually cluster over a small region in the parameter space, i.e., the support for the posterior, 
$\overline{l}(\xi)$ becomes accurate in that region as well, reducing the measure of the inaccurate set $\rho[\mathbbm{1}_{\Xi_e^{\perp}}]$ as a result.
In addition, as the particles become more compact in a local region, 
expensive refinements of the local RB model are less often triggered due to the local accuracy of $\overline{l}(\xi)$.

We first present the adaptive refinement of the local RB surrogate over a given particle set $\{\xi^t_i\}_{i=1}^m$ in Algorithm \ref{ag:gibbsRB}.
We then show the complete adaptive SMC method in Algorithm \ref{ag:gibbsSMC}. 
In Algorithm \ref{ag:gibbsRB}, we note that the accuracy parameter $e_{\text{thre}}$ is prescribed beforehand and can be made adaptive as well.
For example, we can further improve computational efficiency by setting a larger $e_{\text{thre}}$ in the beginning stage of the SMC algorithm and gradually reduce $e_{\text{thre}}$ throughout the iterations.
This strategy leads to computational savings in the beginning stage when the particles are far from the support of the posterior and are less relevant for characterizing the posterior distribution. 
However, it is essential to set $e_{\text{thre}}$ small enough such that each SMC iteration still leads to the particles moving towards the posterior.  

One possible approach is to choose $e_{\text{thre}}$ based on the range of variation of $\overline{l}(\xi)$ over the current particle set $\{\xi^t_i\}_{i=1}^m$, 
e.g., we can set $e_{\text{thre}}$ to be a small fraction of the standard deviation of $\{\overline{l}(\xi^t_i)\}_{i=1}^m$ and compute $e_{\text{thre}}$ automatically for each SMC iteration. 
With this approach, $e_{\text{thre}}$ is large at initial stages of the SMC algorithm where particles are diverse and the range of variation of $\overline{l}(\xi)$ is large.
In the latter stages where particles are more clustered, $\overline{l}(\xi)$ has a smaller range over the particles, which leads to a smaller $e_{\text{thre}}$. 

\RestyleAlgo{boxruled}
\begin{algorithm}[!ht] 
\caption{Adaptive refinement of local RB surrogate \label{ag:gibbsRB} }
\begin{minipage}{.922\linewidth}
Given the current particle set $\Xi_P:=\{\xi^t_i\}_{i=1}^m$, the current local RB surrogate model $\overline{l}(\xi)$ for $l(\xi)$, 
and a desired accuracy threshold $e_{\text{thre}}$,
\begin{itemize}
\setlength\itemsep{0em}
\item  Compute the local RB error indicator $\epsilon_l(\xi)$ for each particle in $\Xi_P$ and $e_{\text{max}} = \max_{\xi\in\Xi_P}\epsilon_l(\xi)$.
 \item \textbf{While} $e_{\text{max}} > e_{\text{thre}}$, \textbf{do}
  \begin{itemize}\itemsep0em
    \item[--] Select the particle $\xi_{\text{max}}=\text{argmax}_{\xi\in\Xi_P}\epsilon_l(\xi)$.
    \item[--] Update $\overline{l}(\xi)$ by the PDE information at $\xi_{\text{max}}$ via local RB method.
    \item[--] Update the error indicators $\epsilon_l(\xi)$ for each particle in $\Xi_P$ and $e_{\text{max}}$.
  \end{itemize}
  \textbf{End}
\item Return the updated surrogate $\overline{l}(\xi)$.
\end{itemize}
\end{minipage}
\end{algorithm}

\RestyleAlgo{boxruled}
\begin{algorithm}[!ht] 
\caption{The adaptive SMC method \label{ag:gibbsSMC} }
\begin{minipage}{.922\linewidth}
Given initial particle approximation $\rho_0^E:=\sum_{i=1}^m w_i^0 \delta(\xi-\xi^{0}_i)$ (and $\overline{w_i^0}:=w_i^0$), the total loss weight $W$ to be applied, set $W_{\text{current}}=0$ and $t=0$.
\begin{itemize}
\setlength\itemsep{0em}
\item \textbf{While} $W_{\text{current}}< W$, \textbf{do}
\begin{itemize}
\setlength\itemsep{0em}
\item Run Algorithm \ref{ag:gibbsRB} to possibly refine the local RB surrogate $\overline{l}(\xi)$ over the current particle set $\{\xi^{t}_i\}_{i=1}^m$.
\item Set $\Delta W = W-W_{\text{current}}$.
 \item \textbf{While} $\text{TRUE}$, \textbf{do}
  \begin{itemize}\itemsep0em
   \item[--] Compute the new weight $\overline{w^{t+1}_i}$ by Equation \eqref{eq:gibbsPosFormulaWts_smc_RB} using $\Delta W$ as the incremental weight.
   \item[--] Compute $\text{ESS}$ of $\left\{\overline{w^{t+1}_i}\right\}_{i=1}^m$.
   \item[--] \textbf{If}  $\text{ESS}>\text{ESS}_{\text{thre}}$, \textbf{break}. 
   \item[--] Backtrack: $\Delta W = \theta \Delta W$.
  \end{itemize}
  \textbf{End}
\item Resample particles $\{\xi^t_i\}_{i=1}^m$ according to $\left\{\overline{w^{t+1}_i}\right\}_{i=1}^m$ to obtain $\{\xi^{t+1,t}_i\}_{i=1}^m$ with weights $1/m$.
\item Mutate each particle in $\{\xi^{t+1,t}_i\}_{i=1}^m$ with Algorithm \ref{al:MCMC_RB} to obtain a new set of evolved particles $\{\xi^{t+1}_i\}_{i=1}^m$,  set $\overline{w^{t+1}_i} = 1/m$, 
	obtain the current particle approximation $\overline{\rho_{t+1}^E}=\sum_{i=1}^m\frac{1}{m} \delta(\xi-\xi^{t+1}_i)$.
\item Set $W_{\text{current}} = W_{\text{current}} + \Delta W$, $t = t+1$.
\end{itemize}
\textbf{End}
\item Report  $\overline{\rho_{t+1}^E}$ as an approximation to the Gibbs posterior $\rho$. 
\end{itemize}
\end{minipage}
\end{algorithm}

\section{Choosing the weight for the loss function} \label{sec:gibbsWt}
In this section we describe one approach to select $W$, the weights in the loss function. The importance of the weights in the Gibbs posterior formulation is to calibrate the loss,
the calibration is automatic in classic Bayesian inference as the density of the data generation process is the calibration. For example, in the case of Gaussian noise
the $\frac{1}{2\sigma^2}$ in the likelihood function can be interpreted as the weight, so when the noise level is high the likelihood is discounted as compared to the low noise setting.
Methods to select $W$ are generally subjective in the context of the Gibbs posterior, and often problem dependent as well \cite{bissiri2016general, syring2016robust}.
The work in \cite{bissiri2016general} introduced several subjective ways to select $W$. In particular, one proposed method is to select $W$ by balancing two isolated loss terms from the objective function \eqref{eq:gibbsOpt}.  In settings where large data sets are available, one can also select $W$ using methods like cross-validation to tune the predictive performance of the posterior.  
 
For inverse problems, however, one typically has access to a rather limited number of observations, so without some assumption on the noise in the data it is hard to quantify
 uncertainty about the inverse solution. We will make some weak assumptions on the noise to provide a method to set the weight. We assume the noise are i.i.d with mean and standard deviation
$$\mathbb{E}[\epsilon] = \epsilon^M, \quad \epsilon^D = \left(\mathbb{E}[\epsilon^2]-\mathbb{E}[\epsilon]^2 \right)^{\frac{1}{2}}.$$
We adopt an approach that is in the same spirit as the Morozov's discrepancy principle \cite{colton1997simple, scherzer1993use}. We select a weight for which the mean and
standard deviation of residual of the posterior predictions will match the statistics on the observed data
\begin{align}\label{eq:gibbsWtObj}
W_{\text{opt}} = \text{argmin}_{W\in \bm{W}} \frac{\|\frac{1}{n}\sum_{i=1}^n \bar{\epsilon_i}(W) - \epsilon^M\| + \|\sqrt{\frac{1}{n-1}\sum_{i=1}^n (\bar{\epsilon_i}(W)- \epsilon^M)^2} - \epsilon^D\| }{\|\epsilon^D\|}
\end{align}
with $\bar{\epsilon_i}(W) = \mathcal{F}(\mathbb{E}_{\rho}(\xi)) -d_i$ is the posterior predicted noise or residual for observation $i$, given weight $W$.
Selecting the set $\bm{W}$ to optimize over in the above equation is nontrivial. In addition, when the sample size is small,  Note that, even with only one observation, if multiple channels are available, one can still estimate the statistics of noise and use \eqref{eq:gibbsWtObj} to select $W$, however one would not have a great deal of trust in the estimate.

We know that the weight would be  $\frac{1}{2{(\epsilon^D)}^2}$ for the classic Bayesian setting with square loss and Gaussian noise. Using this information we provide a discrete 
grid of candidate weight values  $\bm{W}:= \left [\frac{1}{2{(\epsilon^D)}^2T}, T\frac{1}{2{(\epsilon^D)}^2} \right]$, where $T>1$ is a range parameter (e.g., $T=50$). The discretization
is for computational efficiency. To address the case where we may have a very small sample size we stabilize our weight estimate by modeling averaging with the standard Bayesian case
\begin{align}\label{eq:gibbsWt}
W = \frac{S}{S+n-1} \frac{1}{2{(\epsilon^D)}^2} + \frac{n-1}{S+n-1}W_{\text{opt}}
\end{align}
for some $S \geq  1$ (e.g., $S=10$). When $n$ is small, we favor the empirical weight $\frac{1}{2{(\epsilon^D)}^2}$, when $n$ is large and the noise statistics can be computed with good accuracy and we favor the optimized weight $W_{\text{opt}}$, the above can be considered an empirical Bayes procedure. 

We can take advantage of the sequential structure of the SMC procedure to efficiently evaluate the objective in \eqref{eq:gibbsWtObj}  using intermediate computations from the SMC procedure. This allows us to efficiently compute $W_{\text{opt}}$ upon finishing the SMC run and then compute the final weight  by \eqref{eq:gibbsWt},.


It is worth further investigation to see if one can  choose the weight $W$ purely based on the data, instead of imposing additional assumptions of the noise. In addition, it is useful to understand if the use of further information about $\rho(\xi)$ beyond the posterior mean $\mathbb{E}_{\rho}(\xi)$ would help in determining $W$.

\section{Numerical examples}
Now, we present three numerical examples to show the behavior and computational efficiency of our SMC method.
\subsection{1D advection diffusion equation}
In the first example, we consider a 1D advection-diffusion problem. 
We show that our method recovers the usual Bayesian approach when a likelihood function is available and that we use the negative log-likelihood as the weighted loss function.

Let $D=(0,1)$ and consider the following boundary value problem
\begin{subequations}
\begin{align}
  -\nu \frac{\partial^2 u}{\partial x^2}(x, \xi^*)
  +b(x, \xi^*)\frac{\partial u}{\partial x}(x, \xi^*)
    &=f(x), \quad x\in D \\
    u(0,\xi^*)=u(1,\xi^*) &= 0
\end{align}
\end{subequations}
The diffusivity, $\nu$, and source, $f$, are known whereas the advection field, $b$, is a piecewise constant random field parametrized by two unknown parameters $\xi^*_1$ and $\xi^*_2$ as
\begin{equation}
  b(x, \xi^*)= \left[b_1+2\xi^*_1\right]\mathbbm{1}_{[0,0.5)}(x)
    +\left[b_2+2\xi^*_2\right]\mathbbm{1}_{[0.5,1]}(x)
\end{equation}
where $\mathbbm{1}_S(x)$ is one if $x$ is in the set $S$ and is zero otherwise.

We are able to measure the solution at three different locations of $x = [0.1, 0.5, 0.9]$. Our noisy data is hence given by
$$d = \mathcal{D}u + \epsilon$$
where $\mathcal{D}$ is an operator that maps the solution $u(x,\xi^*)$ to the measurement and $\epsilon$ is a noise vector that contains i.i.d entries. 
We assume the noise is drawn from a Gaussian distribution with standard deviation equal to 10\% of the magnitude of the true data. 
In particular, we have $\epsilon^D = 0.173$.
To match the Gaussian likelihood, we use $W=\frac{1}{2{(\epsilon^D)}^2}  = 16.70$ and the  loss
$$l(\xi, d) = \|\mathcal{D}u(x, \xi) - d\|_{l_2}^2.$$

The values of the known parameters are $\nu = 0.1,\ b_1=-0.5,\ b_2=-0.2, \ f(x)=1$, while  the true values of unknown parameters are $\xi^*_1=0.2,\ \xi^*_2=0.7$. 
For the prior distributions, we assume $\xi_1\sim U[0, 1],\ \xi_2\sim U[0, 1]$. 
We use Algorithm \ref{ag:gibbsSMC} to compute the Gibbs posterior with $m=100$ evolving particles and local RB  accuracy set to be $1e-3$. 
In addition, as reference, we perform the standard Random Walk Metropolis-Hastings algorithm 
with Gaussian likelihood to obtain $5,000$ samples from the posterior with $1,000$ burn-in steps. 

We show the comparison of our SMC result with the MCMC reference in Figure \ref{fig:comp_bayesian}. 
Clearly, the SMC method performs similarly to the reference in approximating the posterior distribution. 
In particular, the SMC method took just 3 iterations to reach a good approximation of the posterior. The evolution of the particles, the atoms of the local RB surrogate 
and the intermediate distributions are shown in Figure \ref{fig:gibbs_1dadv_step0} to \ref{fig:gibbs_1dadv_step3} for the various iterations.
As can be seen, as the weight $W$ is progressively increased, the particles cluster around the support of the posterior, 
while simultaneously leading the local RB surrogate to concentrate on the relevant region of the parameter space.

We report the accumulative number of PDE solves at each SMC iteration in Figure \ref{fig:gibbs_1dadv_npde}. 
Most of the computational effort corresponding to the construction of the local RB surrogate is spent the first iteration as the particles move the most towards the posterior support.
Once the local RB becomes accurate over the posterior region, further evolution of the particles rarely triggers the refinement of the surrogate.
The total number of PDE solves to obtain the shown posterior for this example was around 200, representing a significant computational saving over the MCMC method. 

\begin{figure}[!ht]
\centering
\subfloat[CDF plot for $\xi_1$] {\includegraphics[width=0.4\linewidth]{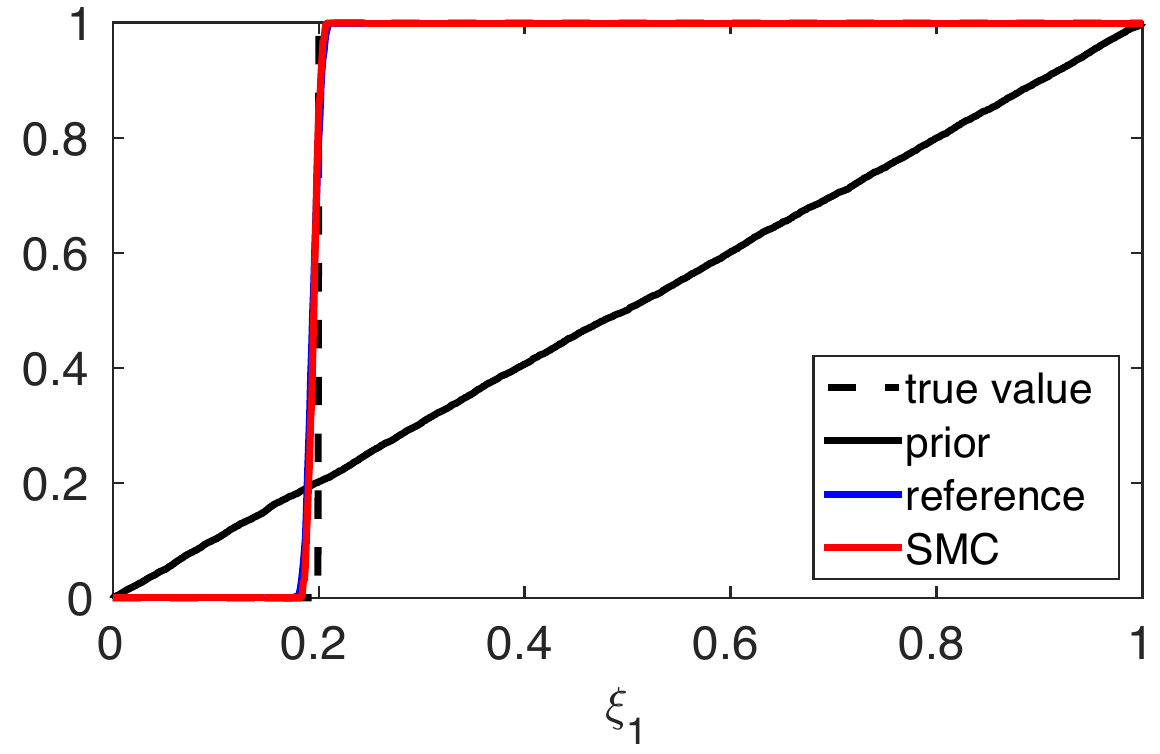}}
\hspace{1cm}
\subfloat[CDF plot for $\xi_2$] {\includegraphics[width=0.4\linewidth]{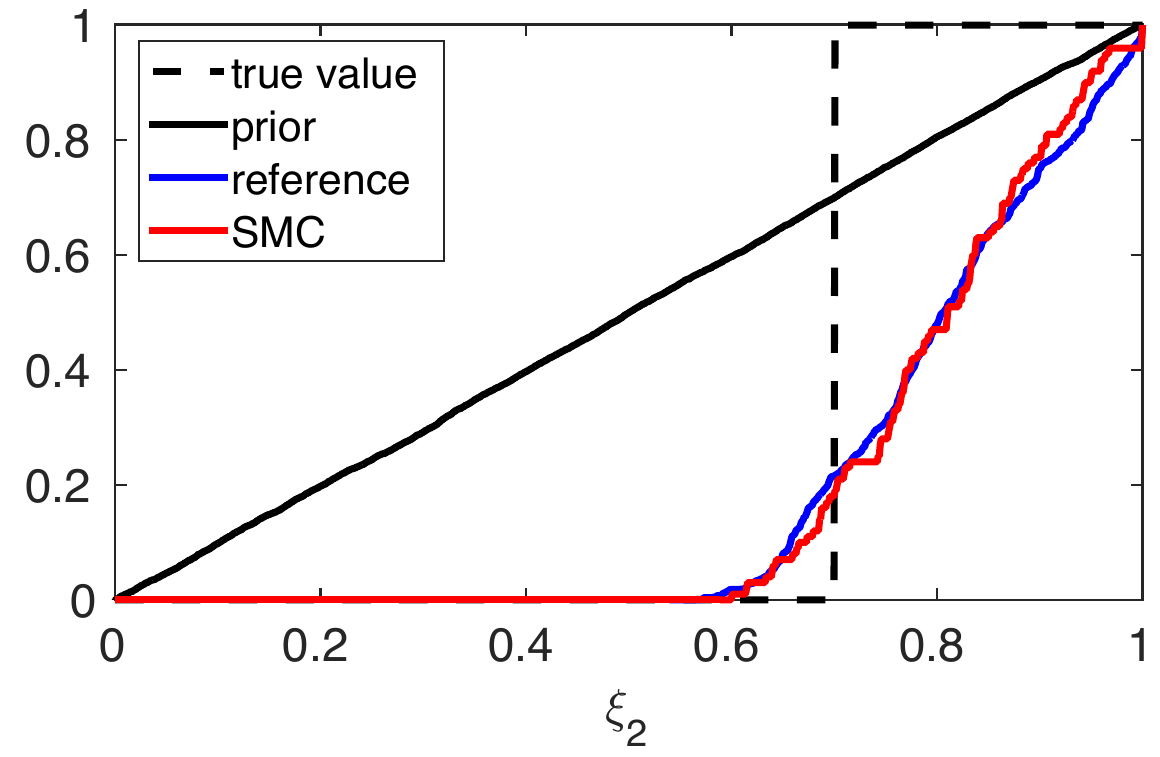}}
\caption[Comparison of the proposed SMC method and the standard MCMC method with Gaussian likelihood function.]
{Comparison of the posterior distribution of the parameters computed by Algorithm \ref{ag:gibbsSMC} and the standard MCMC method with Gaussian likelihood function. }
\label{fig:comp_bayesian}
\end{figure}

\begin{figure}[!ht]
\centering
\subfloat[The true parameters (black), the particles (red) and the atoms for local RB (blue)] {\includegraphics[width=0.27\linewidth]{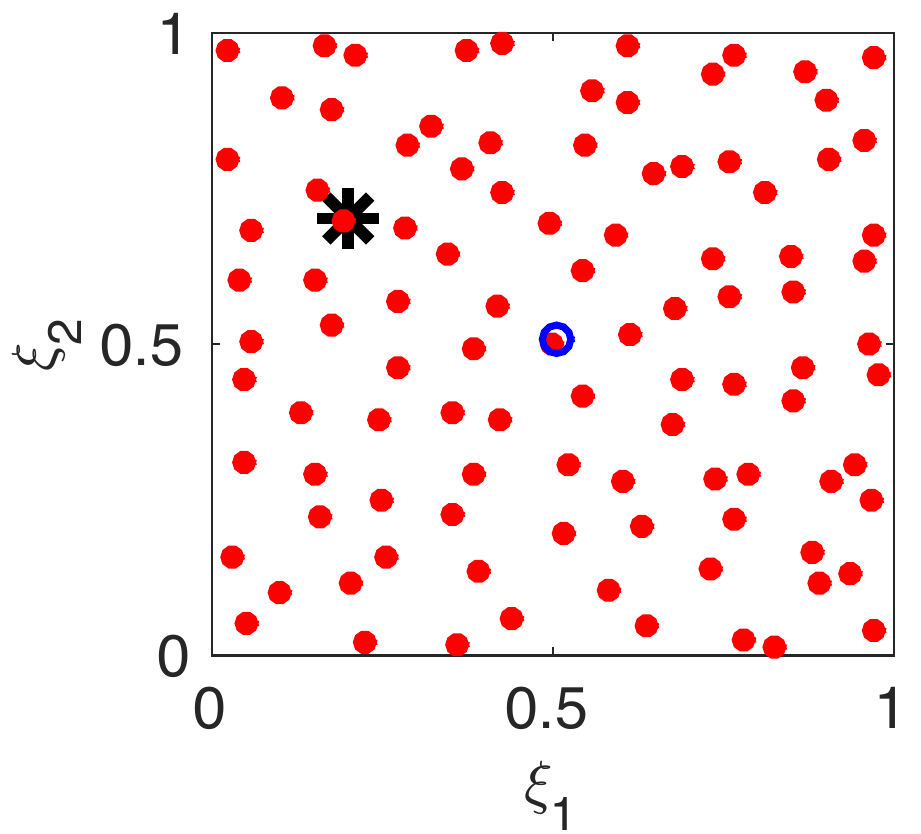}}
\subfloat[CDF plot for $\xi_1$] {\includegraphics[width=0.33\linewidth]{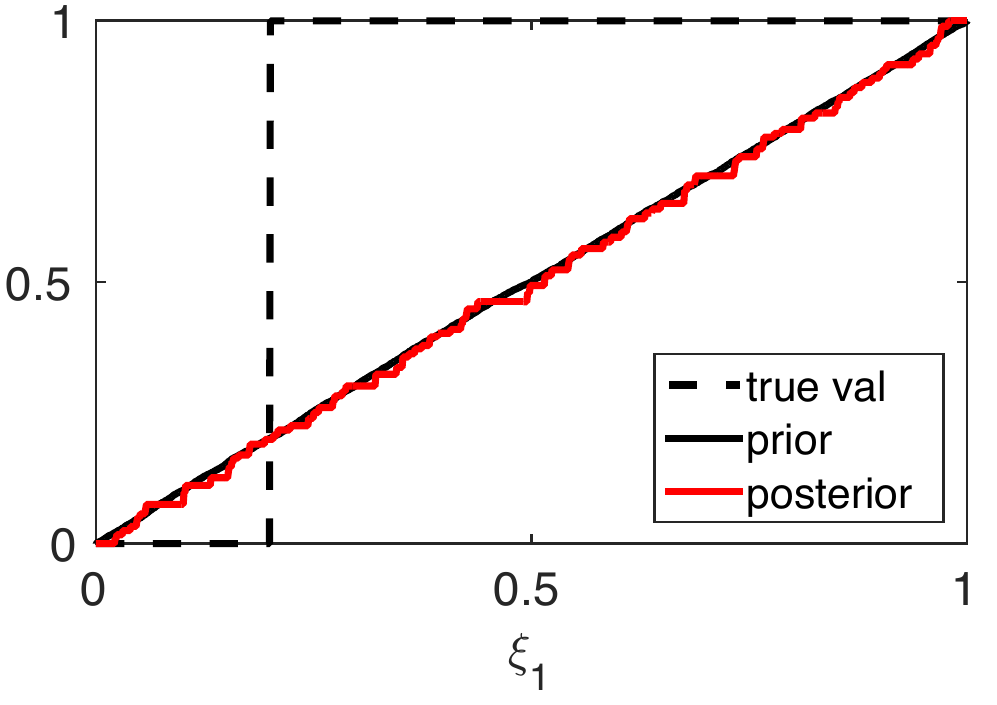}}
\subfloat[CDF plot for $\xi_2$] {\includegraphics[width=0.33\linewidth]{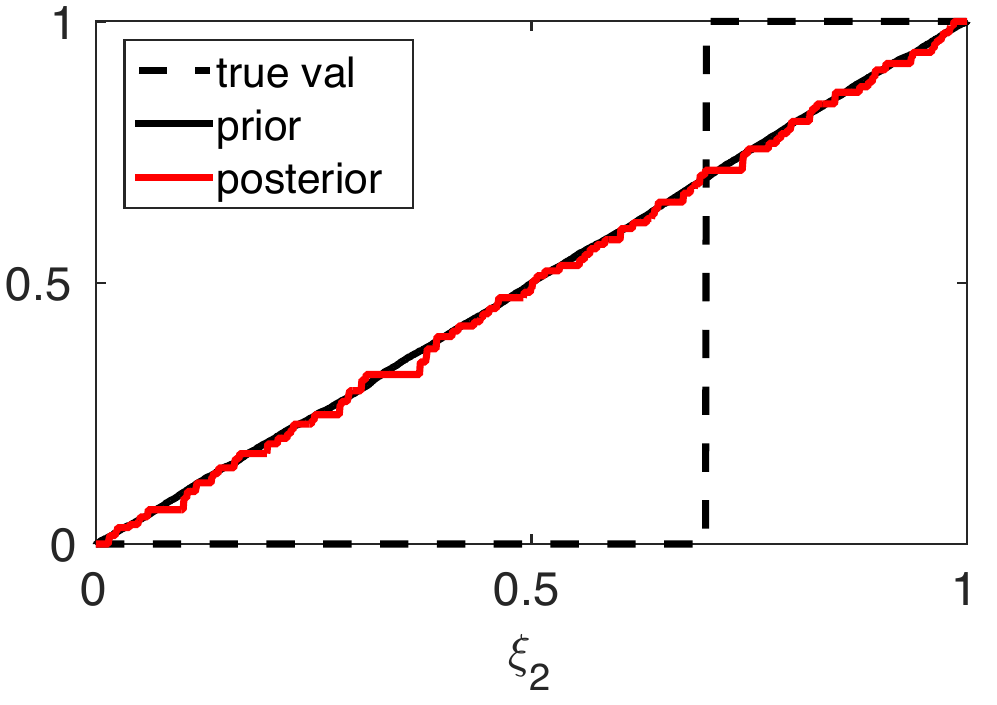}}
\caption{SMC iteration 0, with a loss weight $W_0$ = 0.}
\label{fig:gibbs_1dadv_step0}
\end{figure}
\begin{figure}[!ht]
\centering
\subfloat[The true parameters (black), the particles (red) and the atoms for local RB (blue)] {\includegraphics[width=0.27\linewidth]{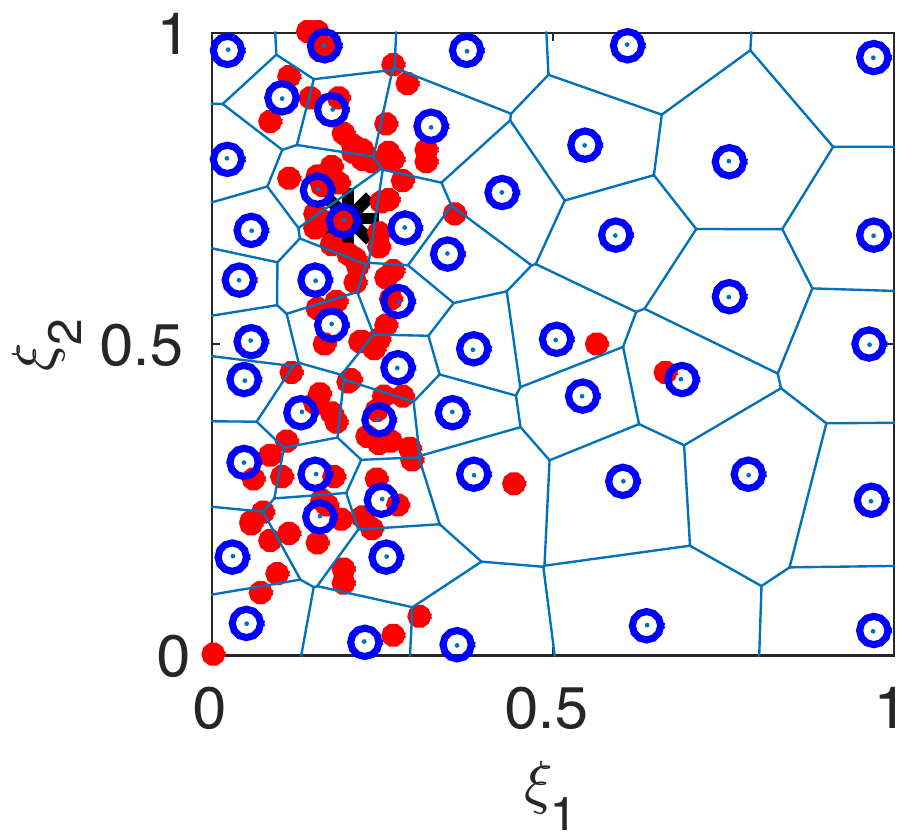}}
\subfloat[CDF plot for $\xi_1$] {\includegraphics[width=0.33\linewidth]{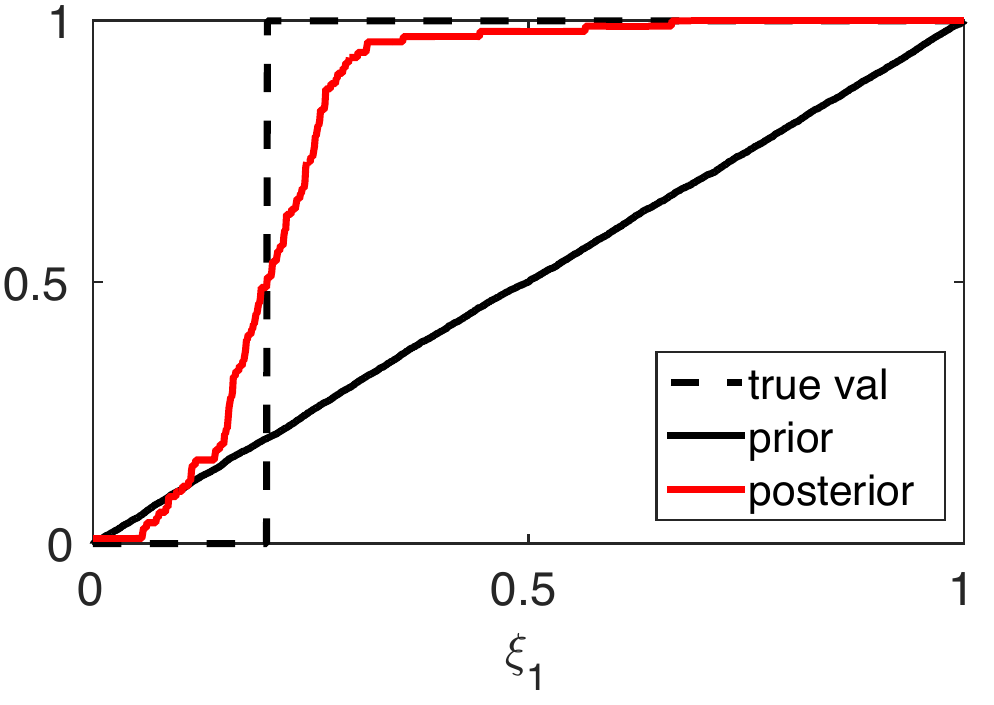}}
\subfloat[CDF plot for $\xi_2$] {\includegraphics[width=0.33\linewidth]{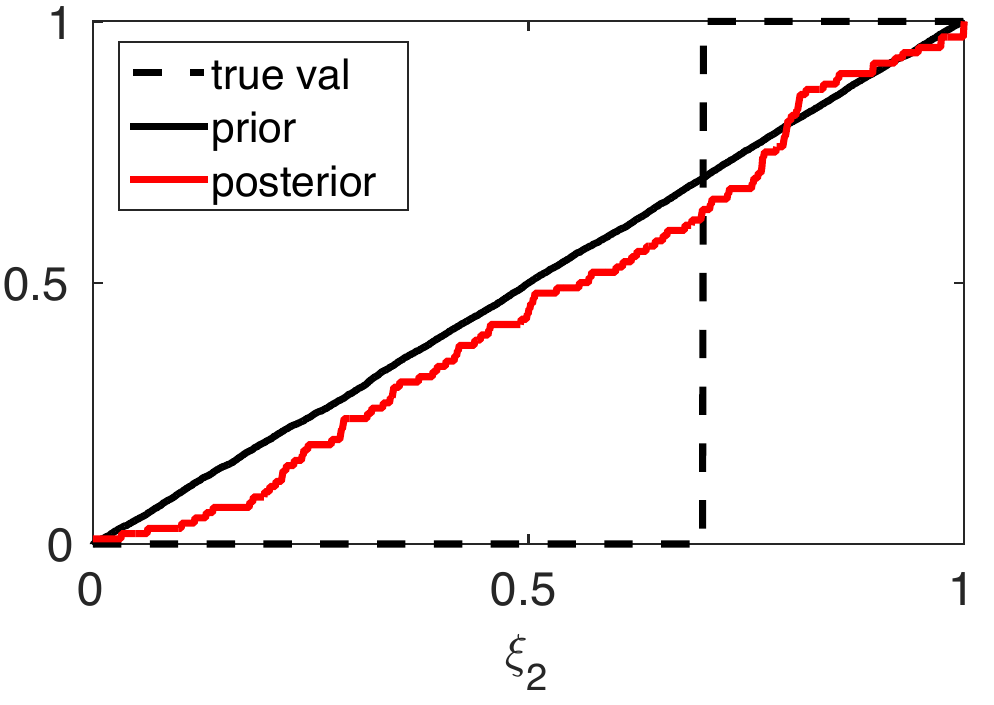}}
\caption{SMC iteration 1, with a loss weight $W_1$ = 0.0782.}
\label{fig:gibbs_1dadv_step1}
\end{figure}
\begin{figure}[!ht]
\centering
\subfloat[The true parameters (black), the particles (red) and the atoms for local RB (blue)] {\includegraphics[width=0.27\linewidth]{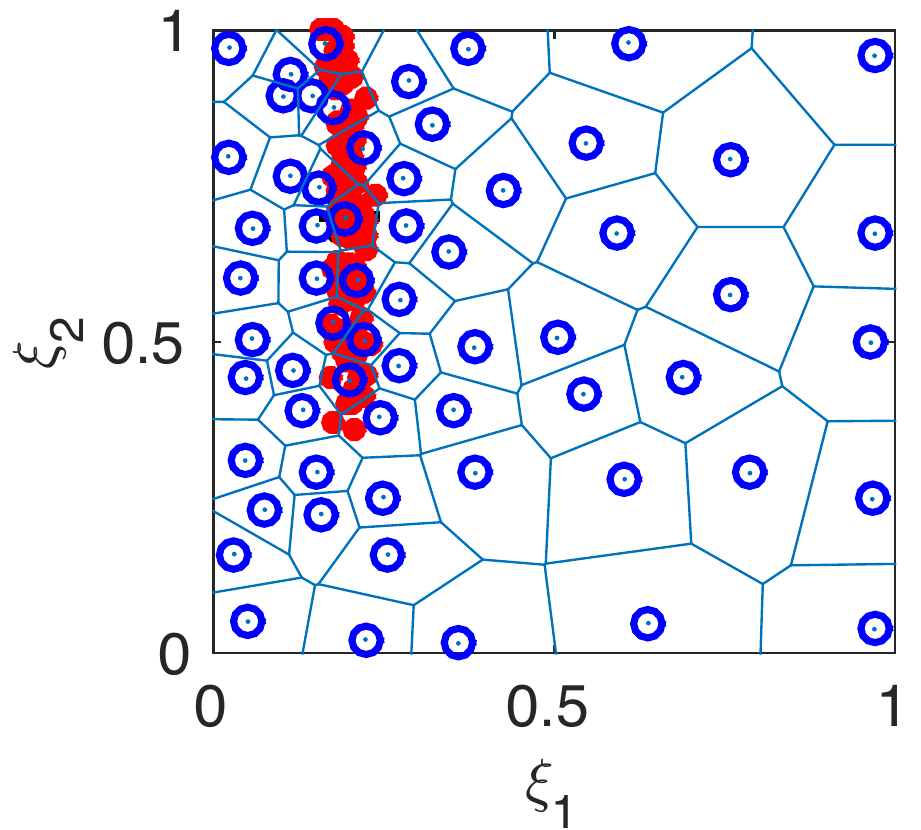}}
\subfloat[CDF plot for $\xi_1$] {\includegraphics[width=0.33\linewidth]{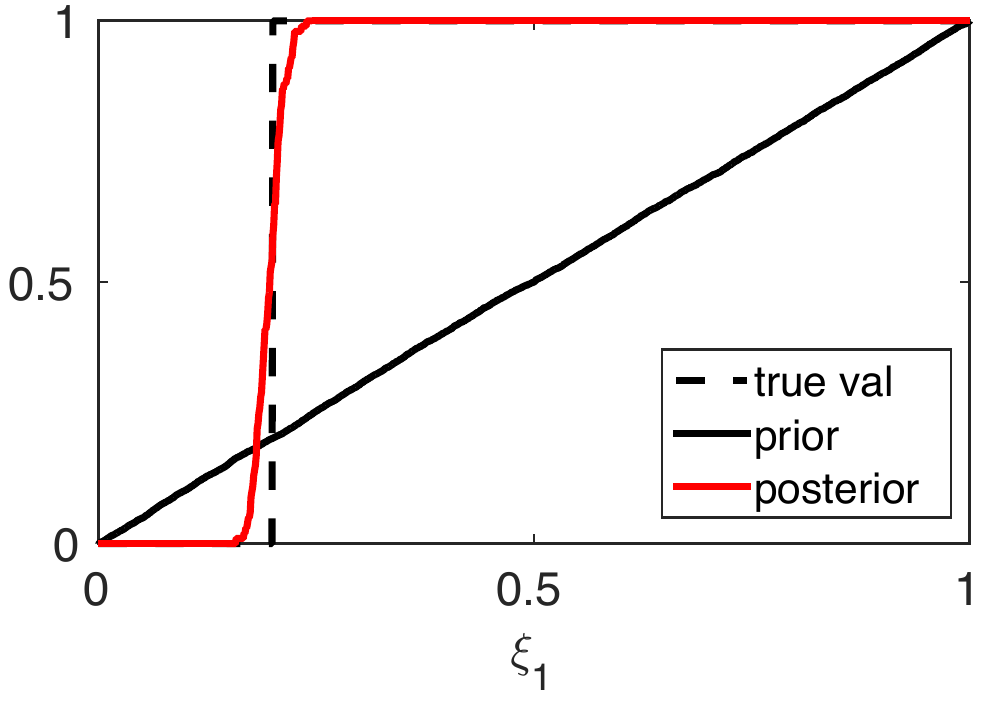}}
\subfloat[CDF plot for $\xi_2$] {\includegraphics[width=0.33\linewidth]{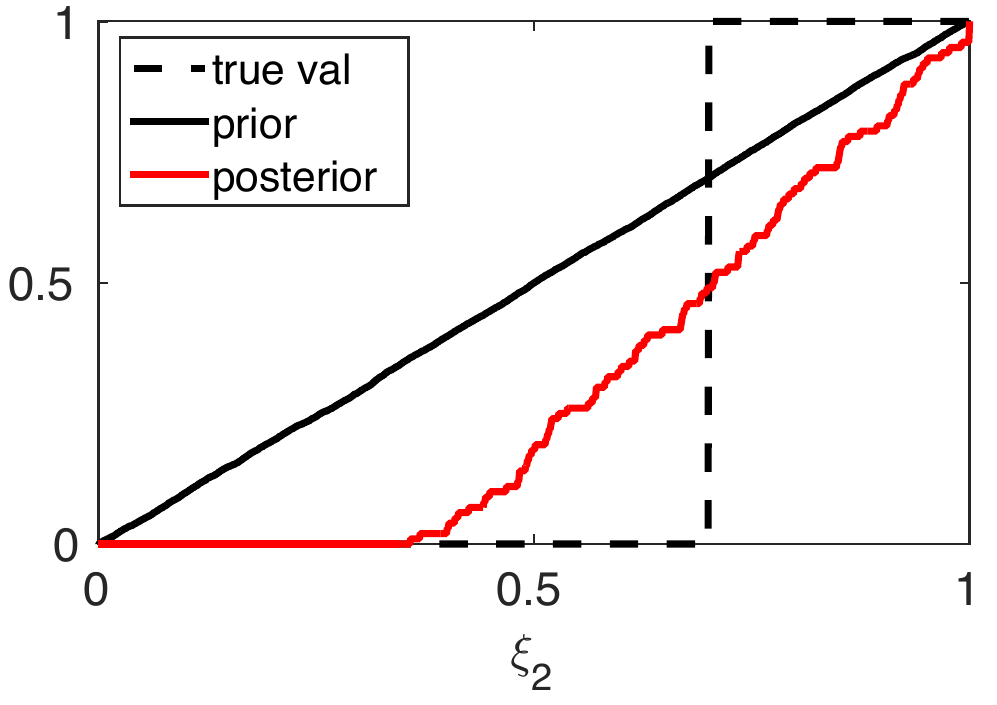}}
\caption{SMC iteration 2, with a loss weight $W_2$ = 1.43.}
\label{fig:gibbs_1dadv_step2}
\end{figure}
\begin{figure}[!ht]
\centering
\subfloat[The true parameters (black), the particles (red) and the atoms for local RB (blue)] {\includegraphics[width=0.27\linewidth]{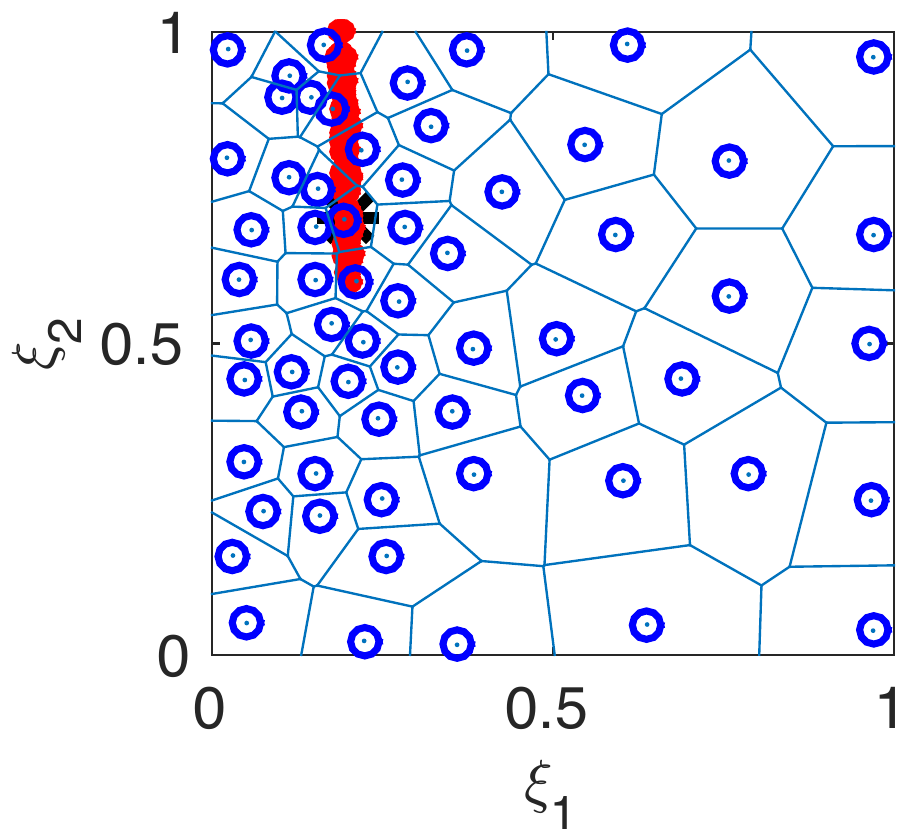}}
\subfloat[CDF plot for $\xi_1$] {\includegraphics[width=0.33\linewidth]{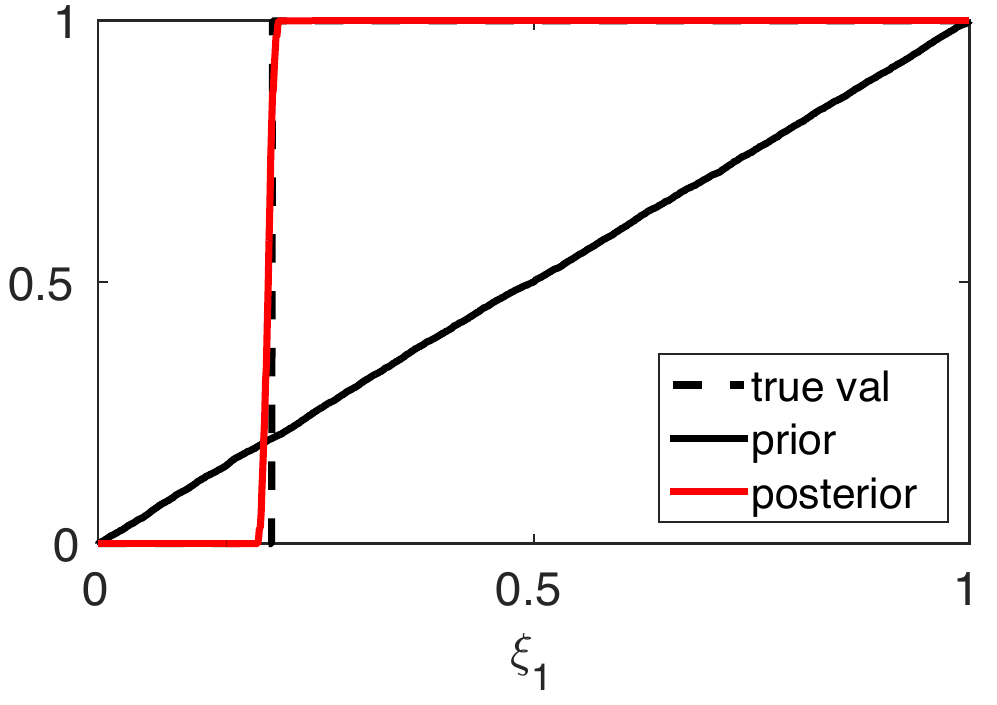}}
\subfloat[CDF plot for $\xi_2$] {\includegraphics[width=0.33\linewidth]{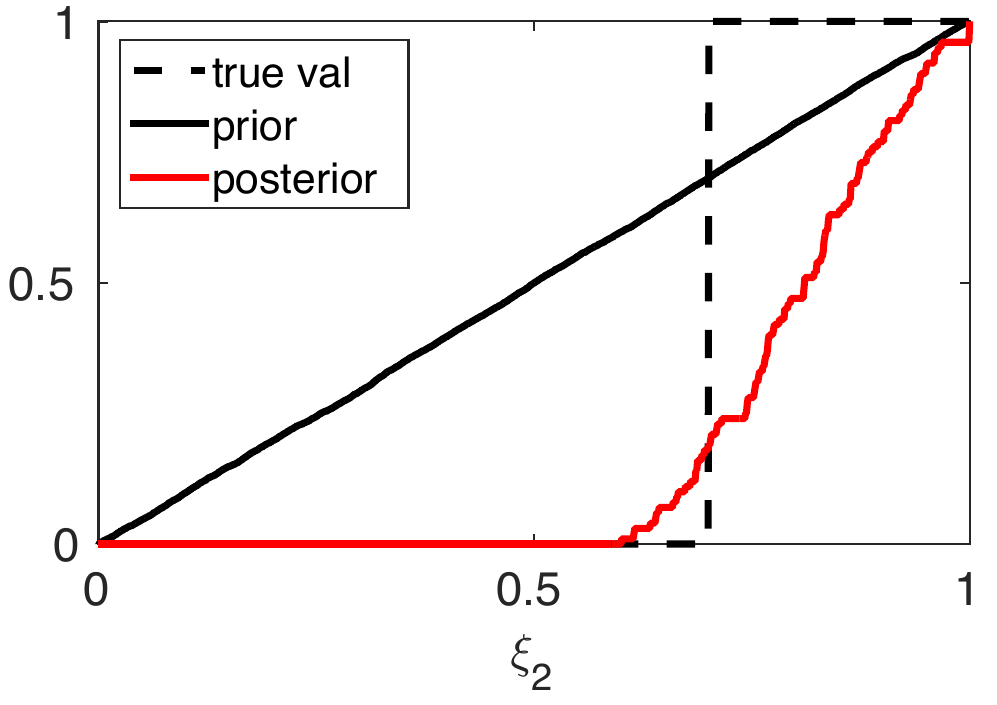}}
\caption{SMC iteration 3, with a loss weight $W_3$ = 16.7.}
\label{fig:gibbs_1dadv_step3}
\end{figure}

\begin{figure}[!ht]
\centering
\includegraphics[width=0.5\linewidth]{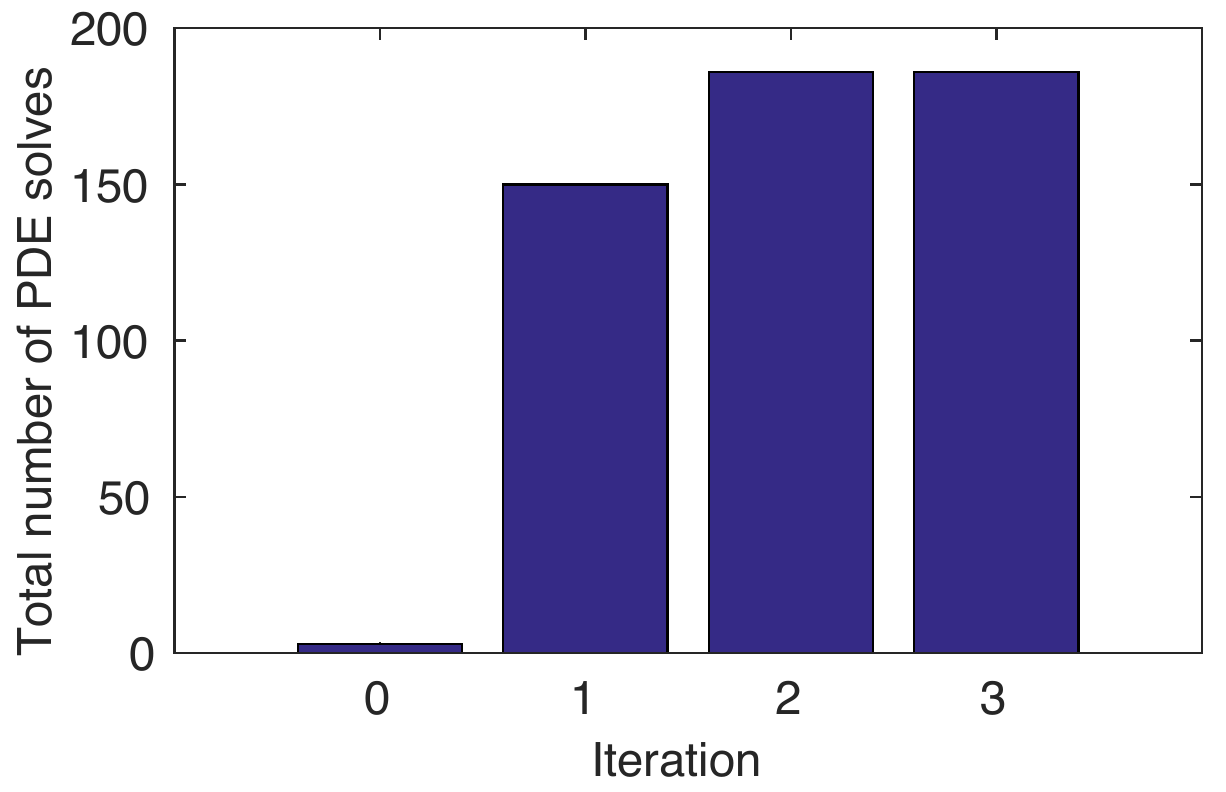}
\caption{The accumulative number of PDE solves at each iteration of the SMC Algorithm \ref{ag:gibbsSMC} for the 1D advection diffusion problem.}
\label{fig:gibbs_1dadv_npde}
\end{figure}

\subsection{2D advection diffusion equation}
In the second example, we consider the simultaneous identification of the diffusivity constant and unknown source for a 2D advection-diffusion problem.
Let $D=(0,1)^2$.  We consider the following problem,
\begin{subequations}
\begin{align}
-\nabla\cdot \left(\kappa(\xi^*)\nabla u(x,\xi^*)\right)+v(x)\cdot\nabla u(x, \xi^*)&=f(x, \xi^*) &&x\in D\\
u(x, \xi^*)&=0 &&x\in\Gamma_d\\
\kappa(\xi^*)\nabla u(x, \xi^*)\cdot n&=0 &&x\in\Gamma_n
\end{align}
\end{subequations}
where $\Gamma_d:=[0,1]\times\{0\}$ and $\Gamma_n:=\partial D\setminus\Gamma_d$.
The unknown parameters $\xi^*$ are included in the diffusivity constant $\kappa(\xi^*)$ and the source term $f(x, \xi^*)$.

In particular, the diffusivity is modeled as
\begin{equation}
\kappa(\xi^*) = 0.02 + 0.98 \xi_1^* 
\end{equation}
The advection field is divergence free and is defined by
\begin{subequations}
\begin{align}
v(x) &= 13 {1\choose 0} + 9 {-x_1\choose x_2}.
\end{align}
\end{subequations}

Finally, the forcing term $f$ is modeled by two Gaussians function with unknown magnitudes, i.e.,
\begin{align}
f(x, \xi^*) =10\exp\left(\frac{-(x_1-0.25)^2-(x_2-0.5)^2}{0.25^2}\right)\xi_2^* + 5\exp\left(\frac{-(x_1-0.75)^2-(x_2-0.75)^2}{0.33^2}\right)\xi_3^*
\end{align}
The goal is to identify $\xi^*$ from noisy measurements of the PDE solution $u(x, \xi^*)$. 
Again, we assume that the concentration field is measured over a uniform grid in the domain. 
Our noisy data is hence given by
$$d = \mathcal{D}u + \epsilon$$
where $\mathcal{D}$ is an operator that maps the solution $u(x,\xi^*)$ to the measurement locations and $\epsilon$ is a noise vector that contains i.i.d entries. Notice that we are assuming that the concentration field has enough regularity as to allow for point-wise evaluations.
We assume the noise is drawn from a Gaussian distribution with standard deviation equal to 20\% of the magnitude of the true data. 
In particular, we have $\epsilon^D = 0.0197$.
For this problem, we use the following $l_1$ loss:
$$l(\xi, d) = \|\mathcal{D}u(x, \xi) - d\|_{l_1}.$$
The weight $W$ was obtained using the approach outlined in Section \ref{sec:gibbsWt}. 
After estimating $W$, we compare the SMC method in Algorithm \ref{ag:gibbsSMC} 
to a Random Walk Metropolis-Hastings method using $\exp(-Wl(\xi))$ as the likelihood. Notice that the Gibbs posterior is invariant with respect to the MC transitions.

The true values of the unknown parameters are $\xi^*_1=0.1,\ \xi^*_2=0.7,\ \xi^*_3=0.5$. 
For the prior, we assume $\xi_1\sim \beta(1, 2),\ \xi_2\sim \beta(3,1),\ \xi_3\sim \beta(3,1)$ and that they are independent. 
The final weight selected was $W=25.8$, representing approximately $1/50$ of $\frac{1}{2{(\epsilon^D)}^2}$. 
We use Algorithm \ref{ag:gibbsSMC} to compute the Gibbs posterior with $m=100$ evolving particles. 
In addition, when training the local RB surrogate model at each SMC step $t$, 
we employ an adaptive accuracy $e_{\text{thre}}$ that is equal to $2\%$ of the standard deviation of $\{\overline{l}(\xi^t_i)\}_{i=1}^m$.
We run the reference MCMC method to obtain $5,000$ samples from the posterior with $1,000$ burn-in steps. 

In Figure \ref{fig:gibb_2dadvSetup}, we show the true diffusivity, advection and source fields. 
In Figure \ref{fig:gibbs_2dadvSol} we show the noise-free PDE solution, the corrupted solution, and the measurement points.
We show the comparison of our SMC result with the MCMC reference in Figure \ref{fig:comp_bayesian_2dadv}. 
Again, the SMC method performs similarly to the reference in approximating the posterior distribution. 
The random variable $\xi_3$ has the largest posterior uncertainty due to the fact that the solution, hence the data, has the least sensitivity with respect to this parameter. Notice that the source associated with $\xi_3$ is located  near the top right corner of the domain and, hence, has limited impact on the concentration at most of the measurement points.

Only 3 iterations of our SMC algorithm were needed to reach the predefined tolerance in this example. Figure \ref{fig:gibbs_2dadv_particles} shows the evolution of particles as well as the local RB atoms throughout the SMC iterations.
Clearly, the particles and local RB atoms simultaneously evolve towards the support of the posterior, leading to an improved approximation of the posterior distribution. 
In addition, as the particles become more clustered, the variation of $\overline{l}(\xi)$ over the particles becomes lower, leading to smaller $e_{\text{thre}}$ (higher accuracy requirement on the surrogate).

Finally, we show the accumulative number of PDE solves at each iteration in Figure \ref{fig:gibbs_2dadv_npde}. 
Note that we did not include the PDE solves in the preprocessing step to select $W$,
and we reinitialized the local RB surrogate model before computing the posterior under the final weight.
We do this to demonstrate how the computational efforts corresponding to the construction of the local RB surrogate are distributed in the SMC iterations. 
The number of PDE solves (local RB atoms) depends critically on $e_{\text{thre}}$ in each iteration. 
In the first iteration, because $e_{\text{thre}}$ is relatively large, only about $100$ PDE solves are incurred. 
In the latter iterations, $e_{\text{thre}}$ becomes lower and the accuracy requirement imposed on $\overline{l}(\xi)$ becomes tighter as well.
On the other hand, the particles become more compact in the latter iterations. 
Though the decreased $e_{\text{thre}}$ demands more PDE solves to refine the surrogate model,
the increased clustering of the particles makes it easier for the local RB surrogate to reach the accuracy requirement. 
These two competing factors jointly determine the number of additional refinements on the surrogates. 
Overall, only less than $400$ PDE solves were incurred in the SMC method to obtain the approximate posterior, 
representing a significant computational saving over the MCMC reference. 

\begin{figure}[!ht]
\centering
\subfloat[Constant diffusivity field] {\includegraphics[width=0.30\linewidth]{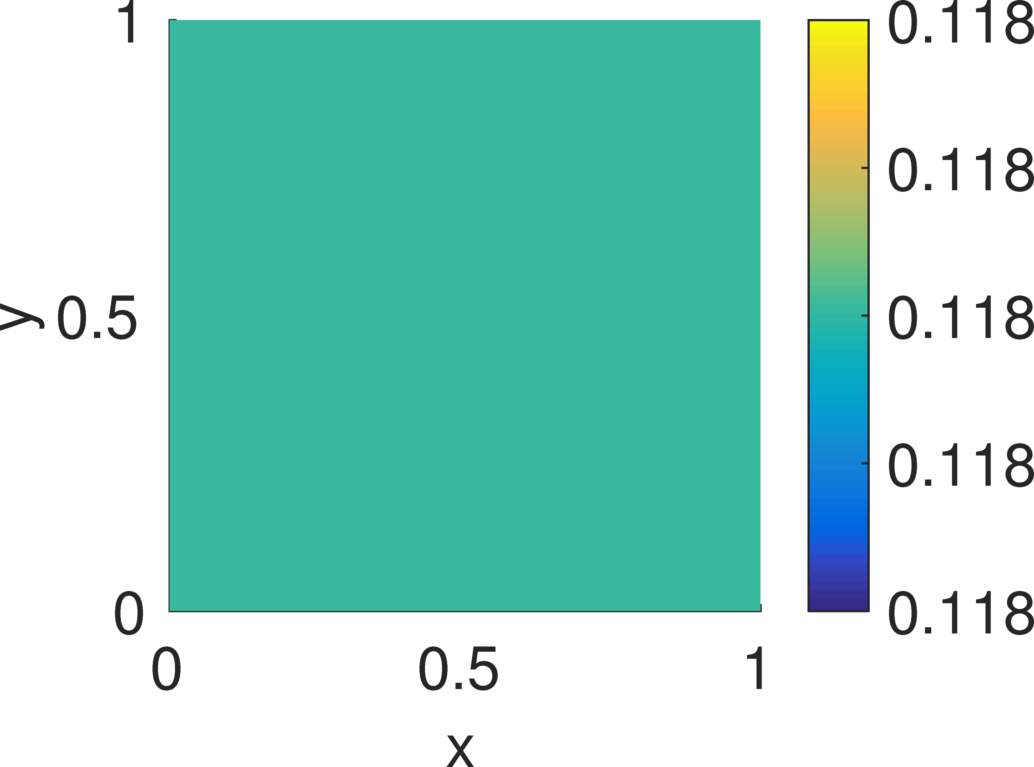}}
\hspace{0.2cm}
\subfloat[Advection field] {\includegraphics[width=0.24\linewidth]{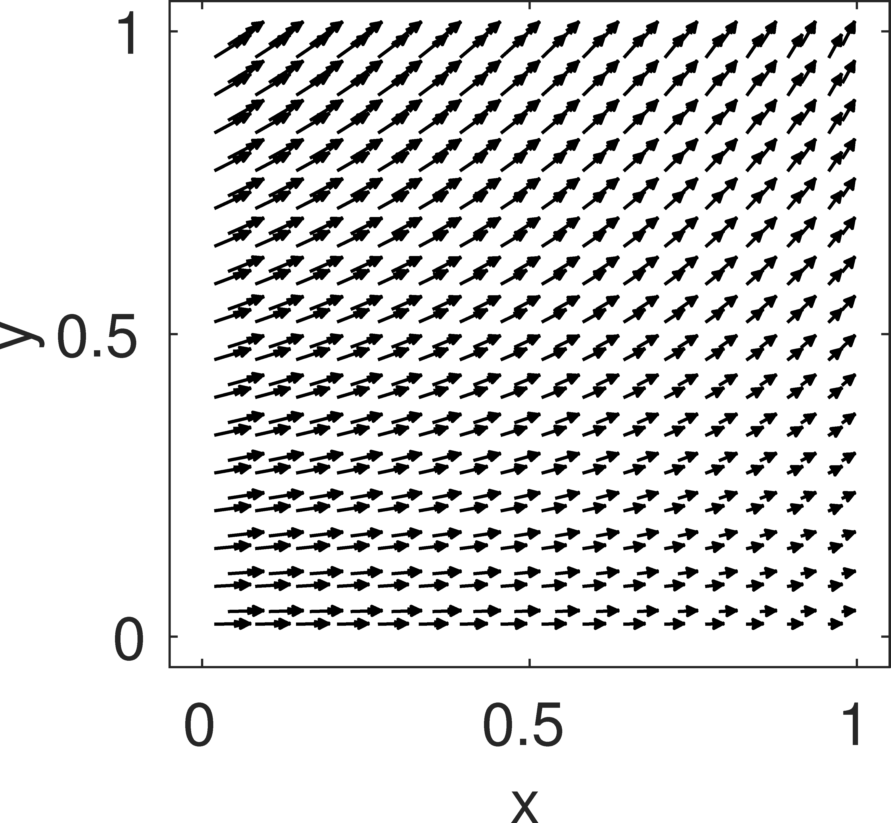}}
\hspace{0.2cm}
\subfloat[Source] {\includegraphics[width=0.27\linewidth]{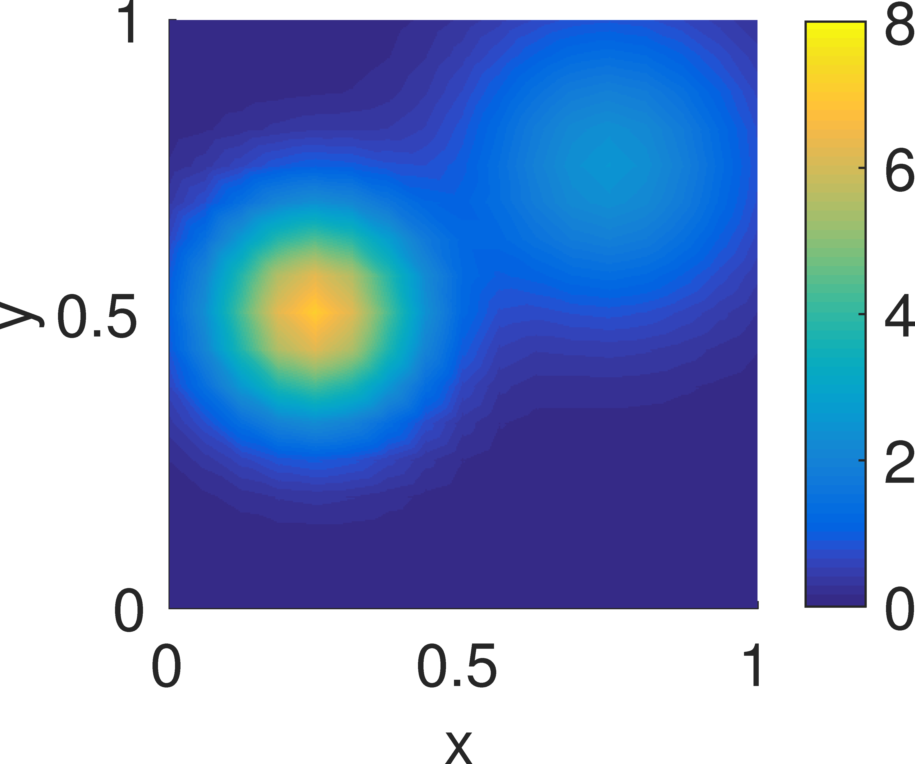}}
\caption{The diffusivity, advection and source fields of the 2D advection-diffusion equation.}
\label{fig:gibb_2dadvSetup}
\end{figure}

\begin{figure}[!ht]
\centering
\subfloat[] {\includegraphics[width=0.4\linewidth]{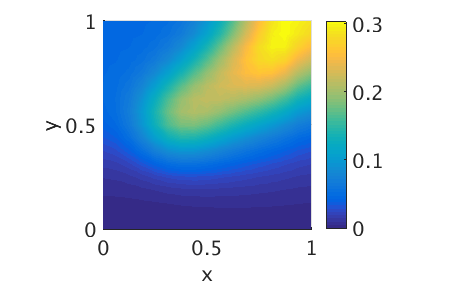}}
\hspace{0.5cm}
\subfloat[] {\includegraphics[width=0.365\linewidth]{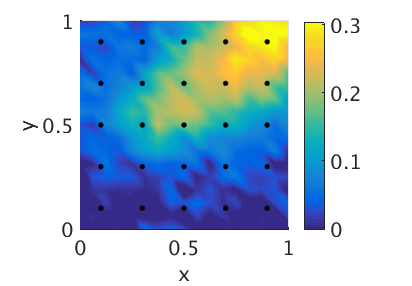}}
\caption{The noise-free solution and the noisy measurements.}
\label{fig:gibbs_2dadvSol}
\end{figure}

\begin{figure}[!ht]
\centering
\subfloat[CDF plot for $\xi_1$] {\includegraphics[width=0.3\linewidth]{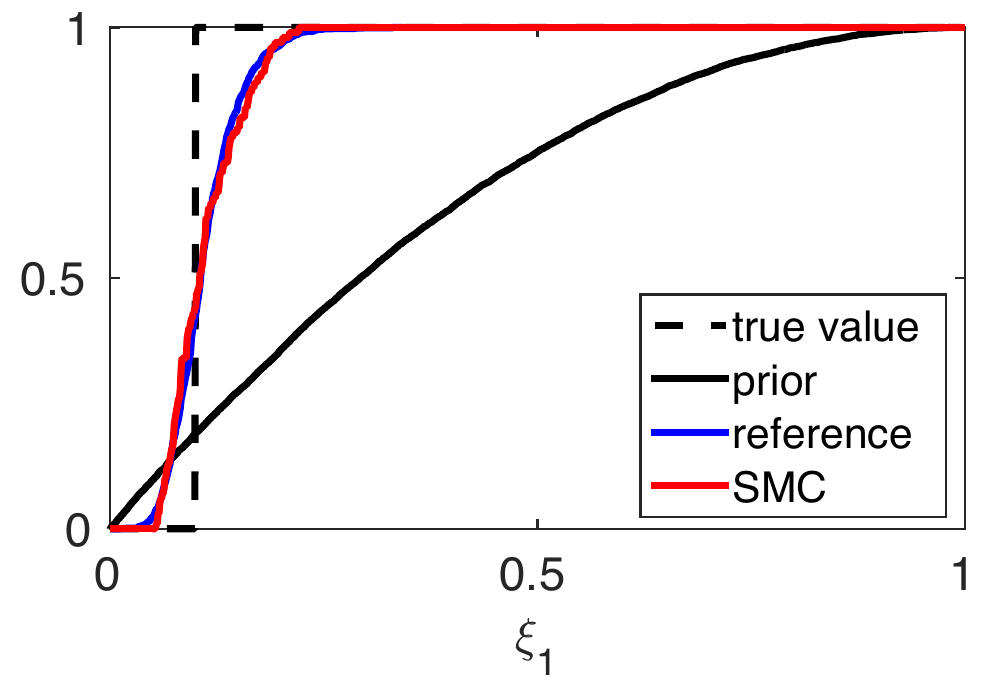}}
\subfloat[CDF plot for $\xi_2$] {\includegraphics[width=0.3\linewidth]{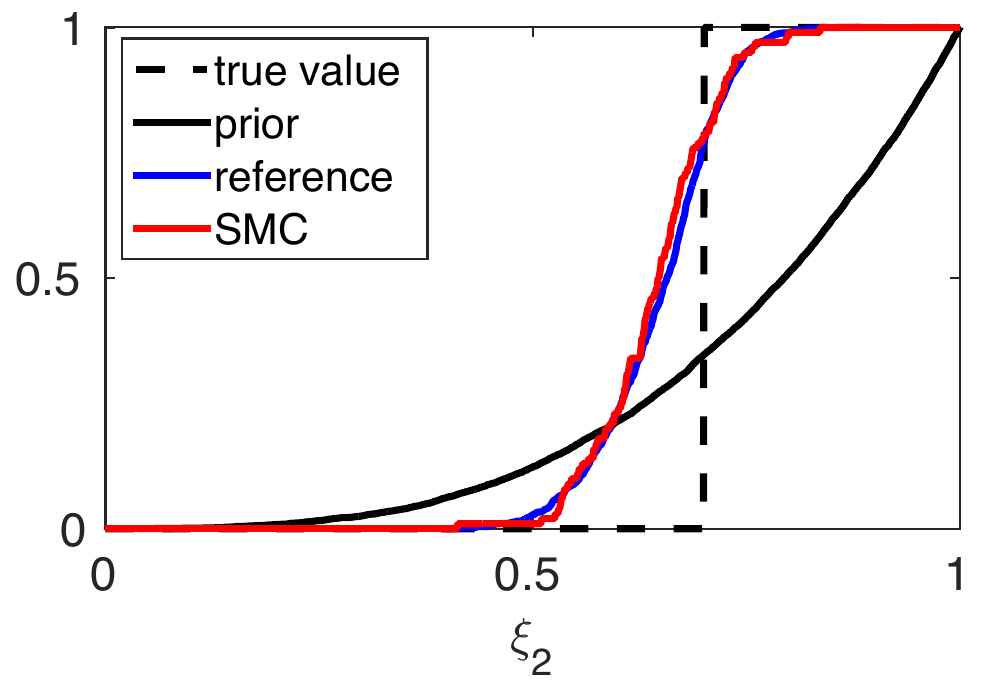}}
\subfloat[CDF plot for $\xi_3$] {\includegraphics[width=0.3\linewidth]{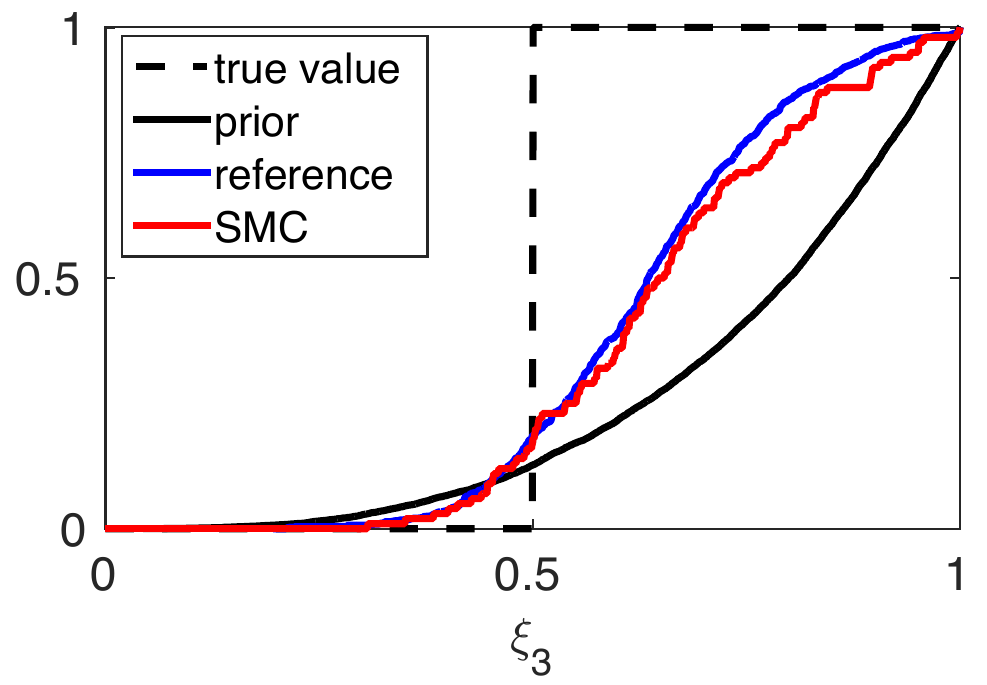}}
\caption[Comparison of the proposed SMC method and the standard MCMC method.]
{Comparison of the posterior distribution of the parameters computed by Algorithm \ref{ag:gibbsSMC} and the standard MCMC method. }
\label{fig:comp_bayesian_2dadv}
\end{figure}

\begin{figure}[!ht]
\centering
\subfloat[$W_0=0$, $e_{\text{thre}}=1.0$]             {\includegraphics[width=0.4\linewidth]{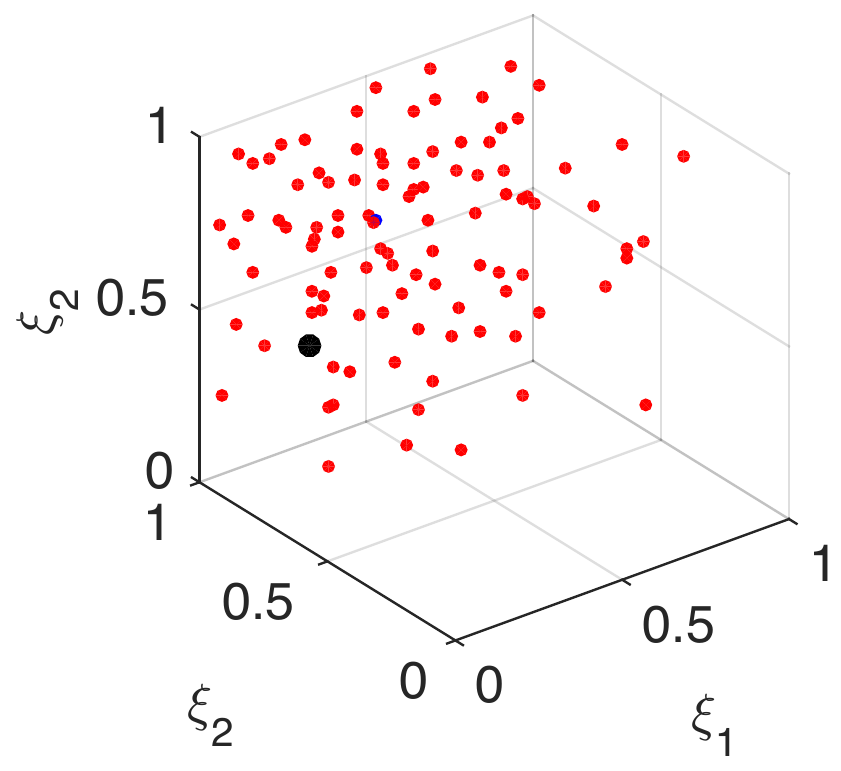}}
\hspace{1cm}
\subfloat[$W_1=4.38$,  $e_{\text{thre}}=0.0163$] {\includegraphics[width=0.4\linewidth]{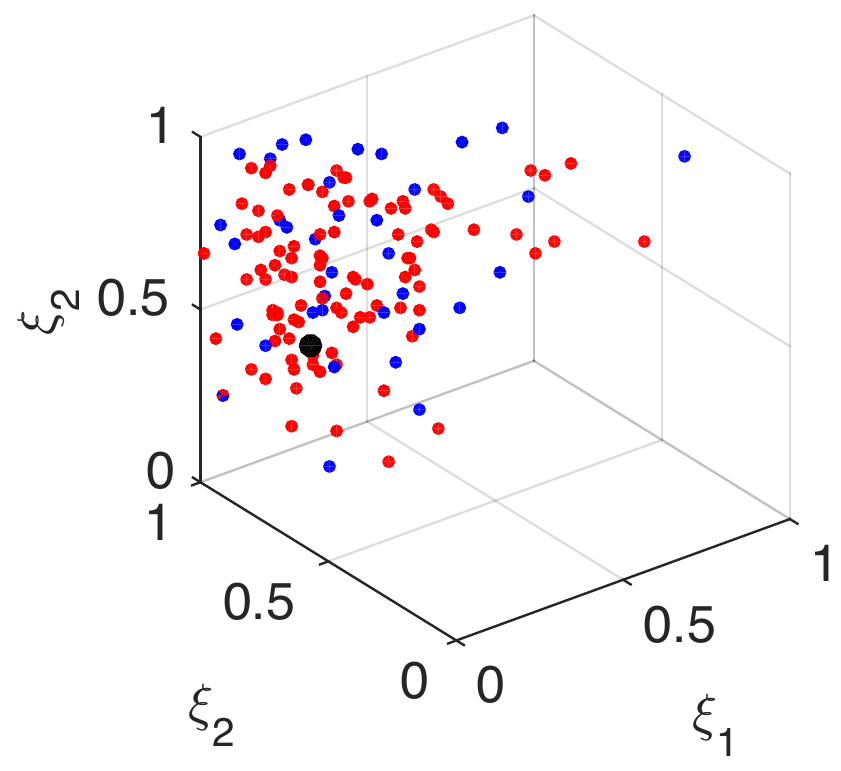}}
\\
\subfloat[$W_2=19.4$,  $e_{\text{thre}}=0.00395$] {\includegraphics[width=0.4\linewidth]{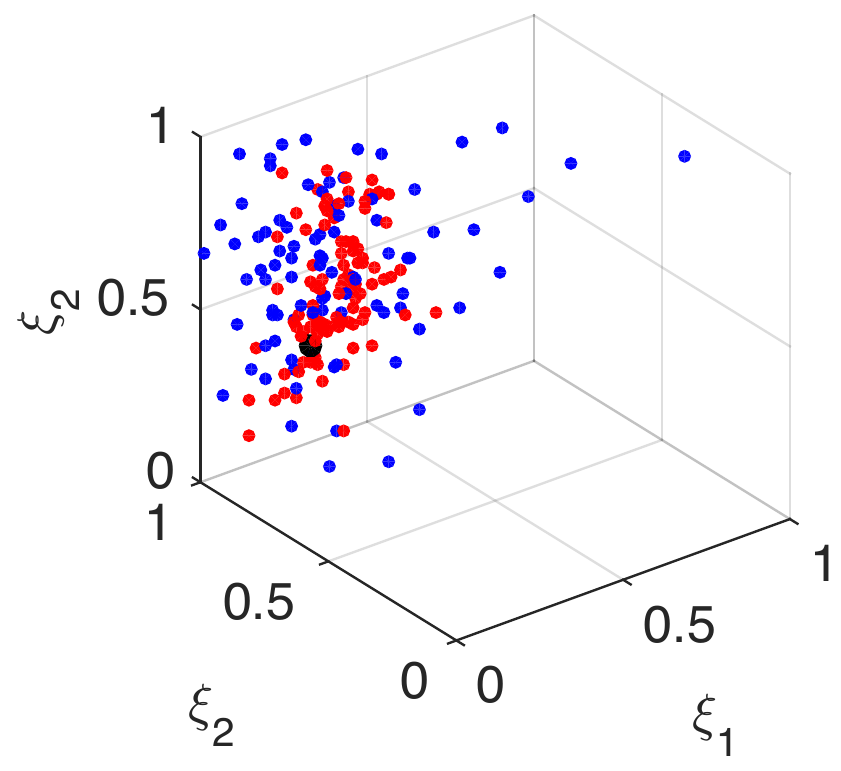}}
\hspace{1cm}
\subfloat[$W_3=25.9$,  $e_{\text{thre}}=0.00188$] {\includegraphics[width=0.4\linewidth]{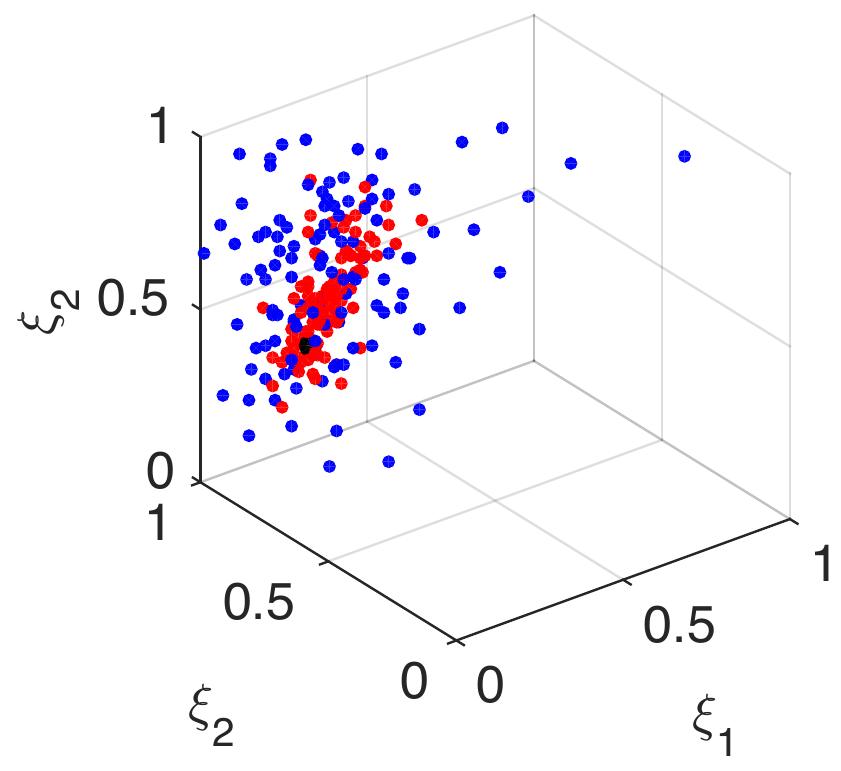}}
\caption[The evolution of the particles and the local RB atoms at each iteration of SMC algorithm.]
{The evolution of the particles and the local RB atoms at each iteration of SMC algorithm. 
The true parameter is in black, the particles are in red and the local RB atoms are in blue. }
\label{fig:gibbs_2dadv_particles}
\end{figure}

\begin{figure}[!ht]
\centering
\includegraphics[width=0.5\linewidth]{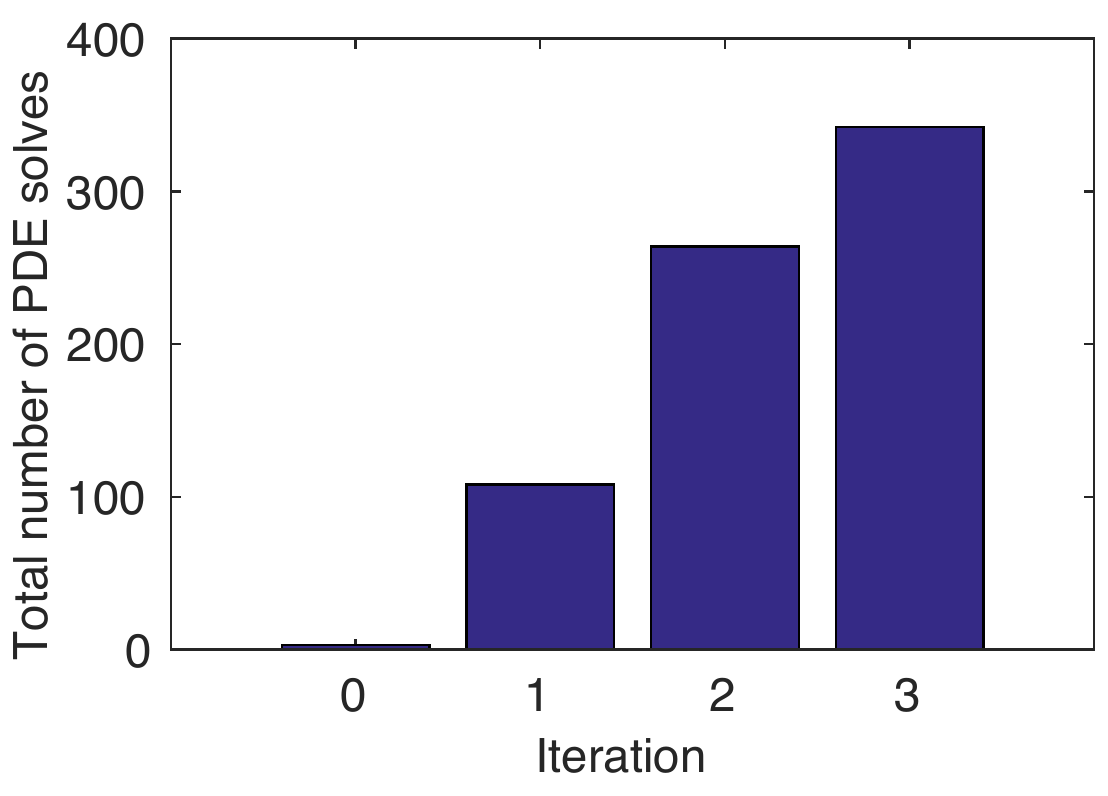}
\caption{The accumulative number of PDE solves at each iteration of the SMC Algorithm \ref{ag:gibbsSMC} for the 2D advection diffusion problem.}
\label{fig:gibbs_2dadv_npde}
\end{figure}

\subsection{2D elasticity equation}
In the last example, we consider two simple elastography problems where we need to infer the distribution of mechanical properties given noisy displacement measurements under known loads. 
These problems are usually characterized by higher dimensionality than those in the previous examples, and hence, are more computationally expensive to solve.

Letting $D=(0,1)^2$,  we consider the following linear elasticity problem.
\begin{subequations}
\begin{align}
-\nabla\cdot \sigma(x, \xi^*) + f  &= 0, &&x\in D\\
\epsilon(x, \xi^*) &= \frac{1}{2}(\nabla u(x, \xi^*) + \nabla u(x, \xi^*) ^T), &&x\in D\\
\sigma(x, \xi^*) &= C(\xi^*):\epsilon(x, \xi^*), &&x\in D\\
u(x, \xi^*)&=0 &&x\in\Gamma_d\\
\sigma(x, \xi^*) \cdot n&=\tau &&x\in\Gamma_n
\end{align}
\end{subequations}
where $\Gamma_d:=[0,1]\times\{0\}$ and $\Gamma_n:=\partial D\setminus\Gamma_d$.
The unknown parameters $\xi^*$ are included in the modulus of the material, which is part of the elasticity tensor $C(\xi^*)$. 
We consider isotropic plane stress problems where we know the Poisson's ratio $\nu=0.3$ 
and try to identify the unknown Young's modulus $E(\xi^*)$ from noisy measurements of $u(x, \xi^*)$.
The setup of the problem as well as the two modulus models we used in this example are shown in Figure \ref{fig:gibb_elasticSetup}.
The two problems have parameter dimensions of 5 and 9, respectively.

\begin{figure}[!ht]
\centering
\subfloat[Boundary conditions.] {\includegraphics[width=0.25\linewidth]{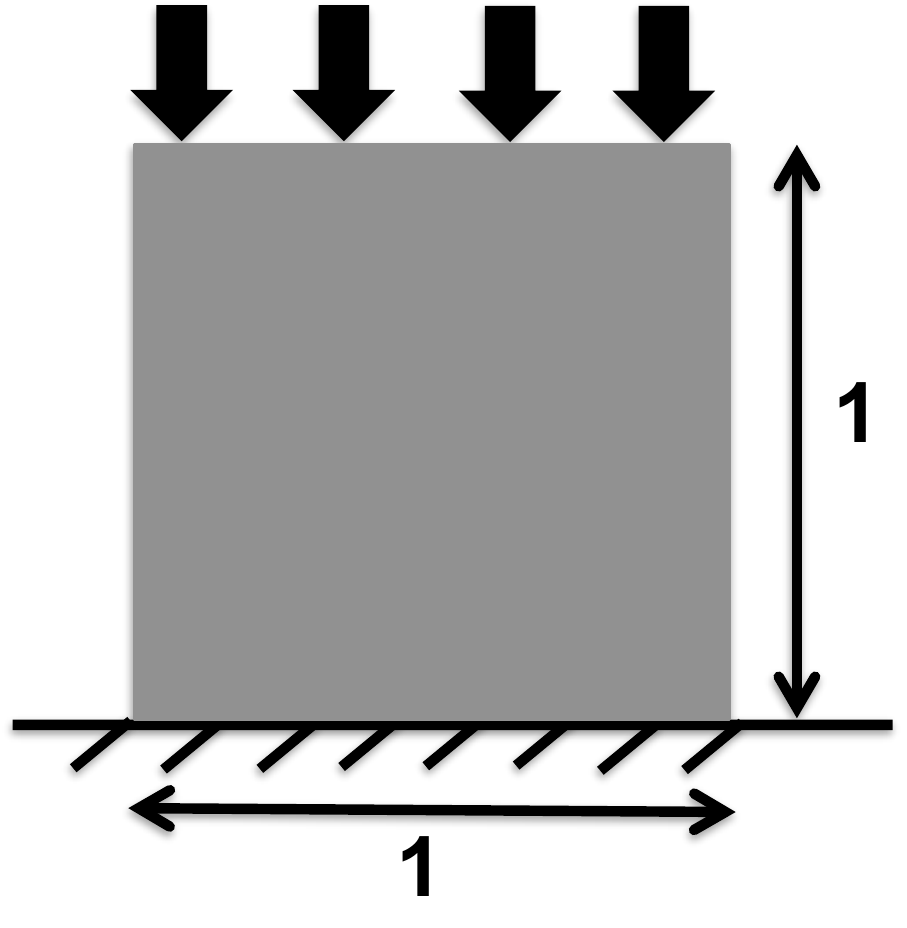}}
\hspace{0.2cm}
\subfloat[Material with layered Young's modulus.] {\includegraphics[width=0.28\linewidth]{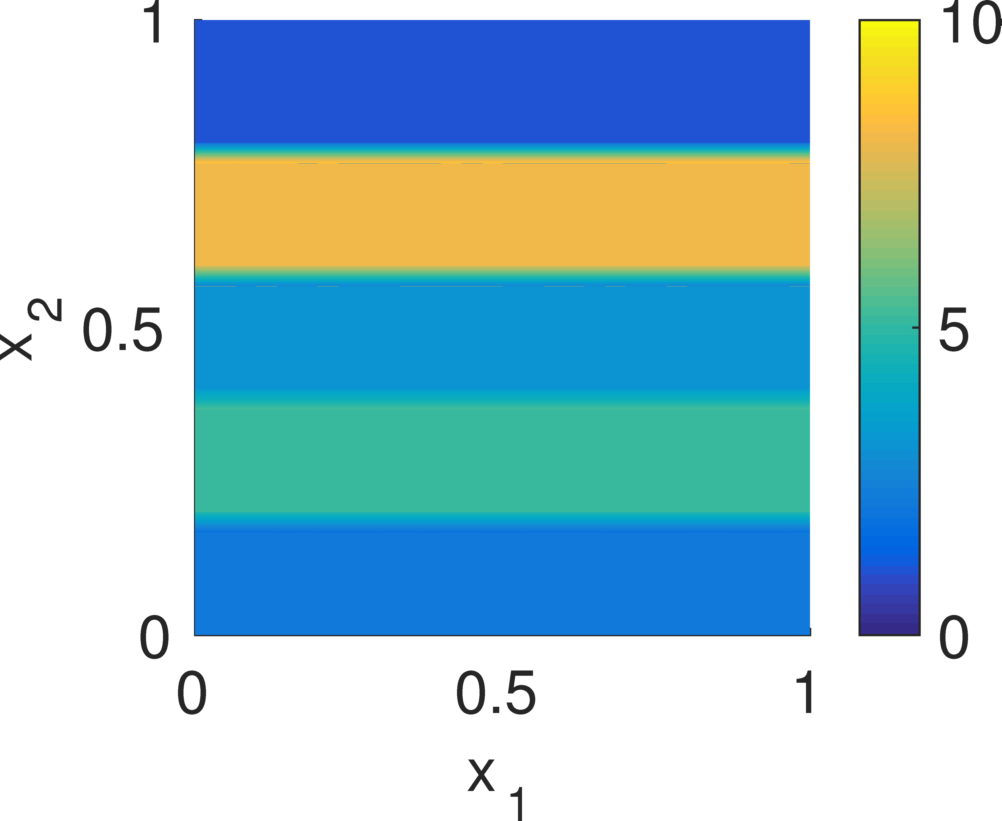}}
\hspace{0.2cm}
\subfloat[Material with hard inclusion. ] {\includegraphics[width=0.27\linewidth]{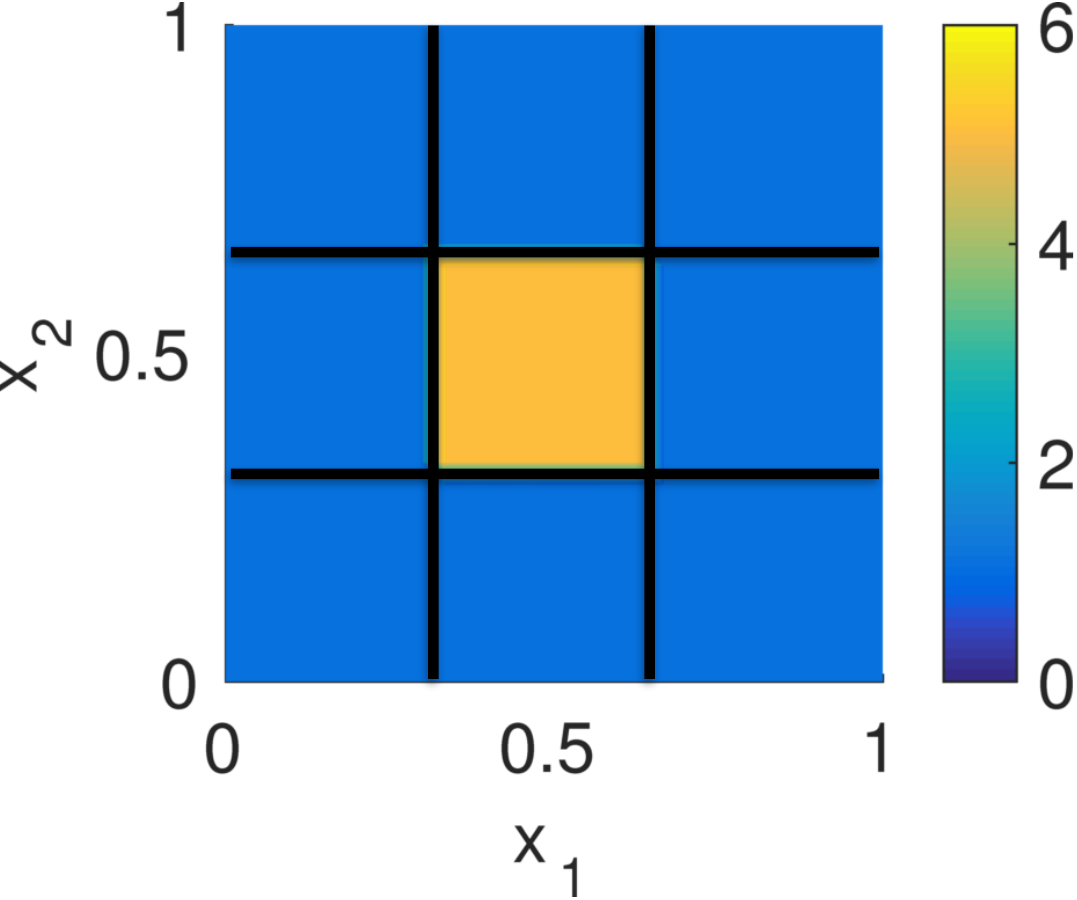}}
\caption{The boundary condition, and true material properties for the simple elastography problems.}
\label{fig:gibb_elasticSetup}
\end{figure}

The displacement fields and the noisy measurements of the two models are shown in Figure \ref{fig:gibbs_layered_sol} and \ref{fig:gibbs_inclusion_sol}, respectively.
Note that we only use the noisy displacement data in the vertical, i.e., $x_2$, direction. We perturb the solution with 5\%, 10\% and 20\% Gaussian noise and investigate the 
posterior mean and deviation computed by the SCM method. 
For the prior distribution, we assume the parameters are independent and follow a $\beta(1,3)$ distribution scaled to the range of $[0.1, 10]$. 
We use a simple $l_2$ loss function and weights $W = \frac{1}{2{(\epsilon^D)}^2}$, which corresponds to the Gaussian noise model exactly.
In addition, when training the local RB surrogate model at each SMC step $t$, 
we employ an adaptive accuracy $e_{\text{thre}}$ that is equal to $5\%$ of the standard deviation of $\{\overline{l}(\xi^t_i)\}_{i=1}^m$.

We plot the inversion for the two models under different level of noise in Figure \ref{fig:gibbs_layered_inversion} and \ref{fig:gibbs_inclusion_inversion}, respectively. 
The posterior mean gives reasonable approximations to the true modulus, and as the level of noise in the data increases, we have higher uncertainty about our inverse solution,  as expected.
Finally, we present the number of local RB atoms used in each of the models with each level of noise in Figure \ref{fig:gibbs_elasticity_natom}.
When the noise level is low, i.e., the weight for the loss is large, 
the posterior becomes increasingly concentrated in a small region within the support of the prior, 
and more SMC steps are needed to approximate the Gibbs posterior, leading to a larger number of refinements (atoms) for the local RB.
Also noticed from the comparison is that when the dimension of the parameter space becomes high, as for the inclusion problem, 
higher computational cost is required to approximate the Gibbs posterior.

\begin{figure}[!ht]
\centering
\subfloat[Horizontal displacement] {\includegraphics[width=0.3\linewidth]{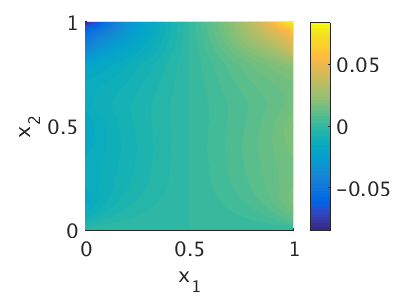}}
\hspace{0.1cm}
\subfloat[Vertical displacement] {\includegraphics[width=0.3\linewidth]{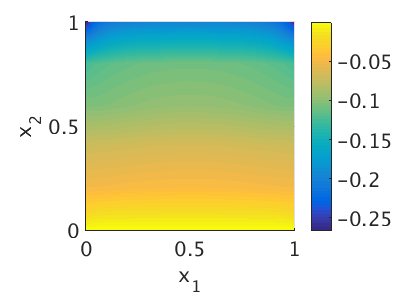}}
\hspace{0.1cm}
\subfloat[Measurement of the vertical displacement with 10\% Gaussian noise] {\includegraphics[width=0.3\linewidth]{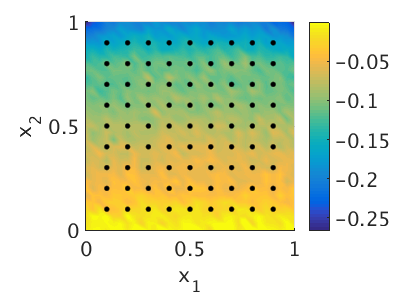}}
\caption{The displacement fields and the noisy measurements for the layered material. }
\label{fig:gibbs_layered_sol}
\end{figure}

\begin{figure}[!ht]
\centering
\subfloat[Horizontal displacement] {\includegraphics[width=0.31\linewidth]{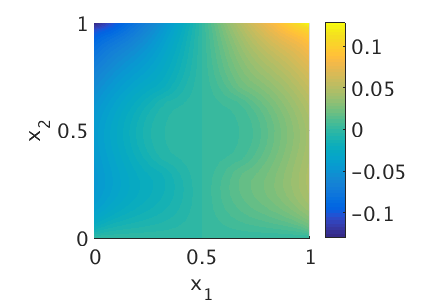}}
\hspace{0.1cm}
\subfloat[Vertical displacement] {\includegraphics[width=0.31\linewidth]{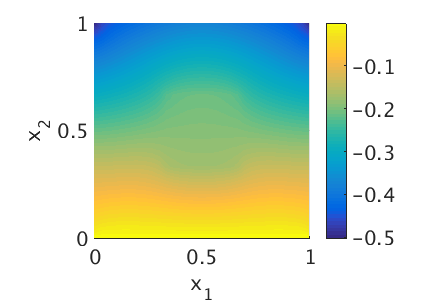}}
\hspace{0.1cm}
\subfloat[Measurement of the vertical displacement with 10\% Gaussian noise] {\includegraphics[width=0.31\linewidth]{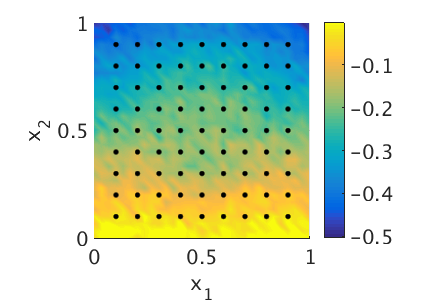}}
\caption{The displacement fields and the noisy measurements for the material with a hard inclusion. }
\label{fig:gibbs_inclusion_sol}
\end{figure}

\begin{figure}[!ht]
\centering
\subfloat[5\% noise] {\includegraphics[width=0.3\linewidth]{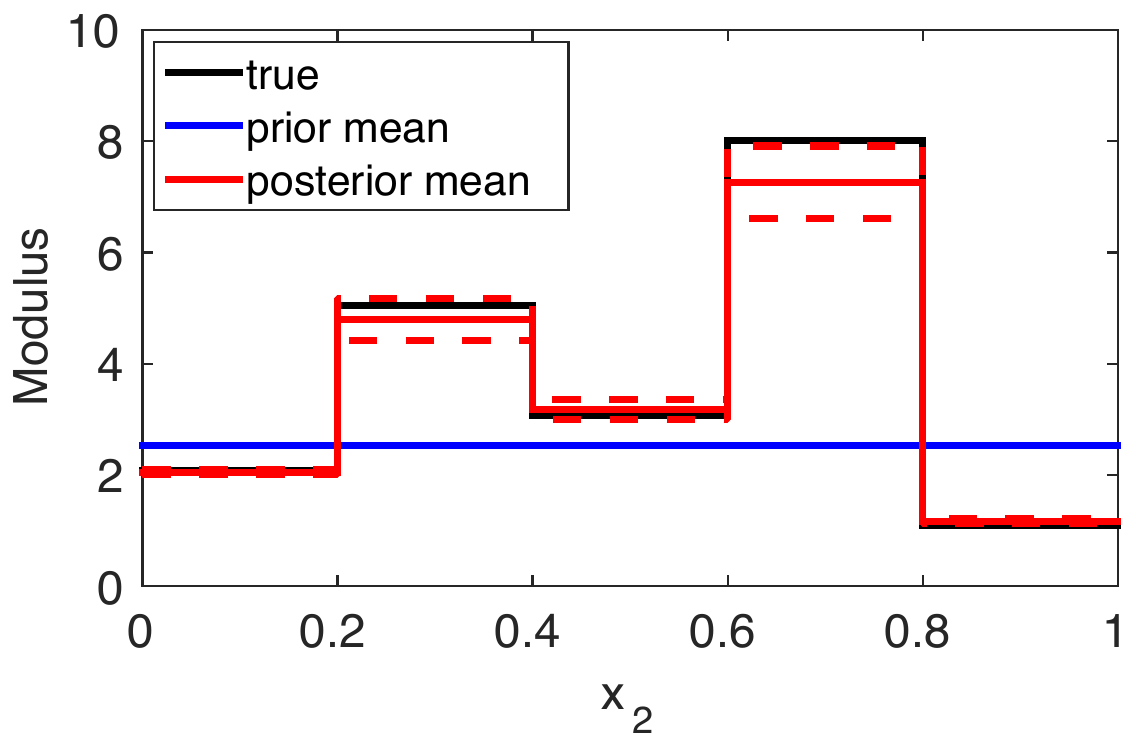}}
\hspace{0.1cm}
\subfloat[10\% noise] {\includegraphics[width=0.3\linewidth]{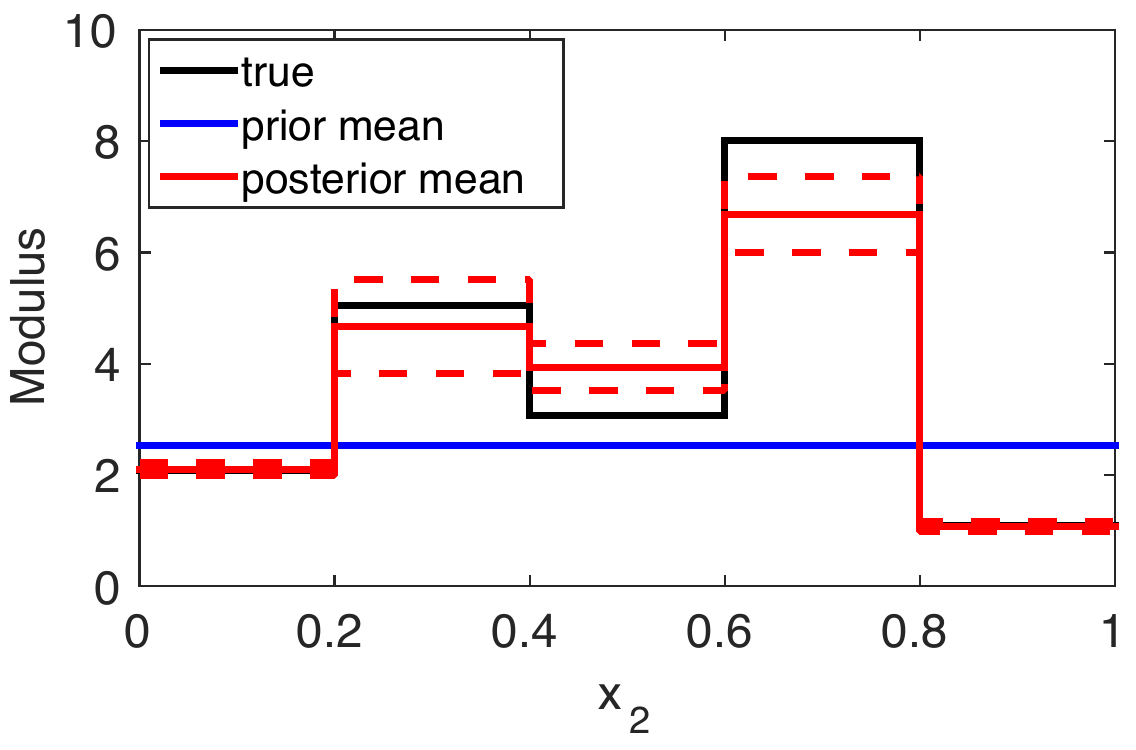}}
\hspace{0.1cm}
\subfloat[20\% noise] {\includegraphics[width=0.3\linewidth]{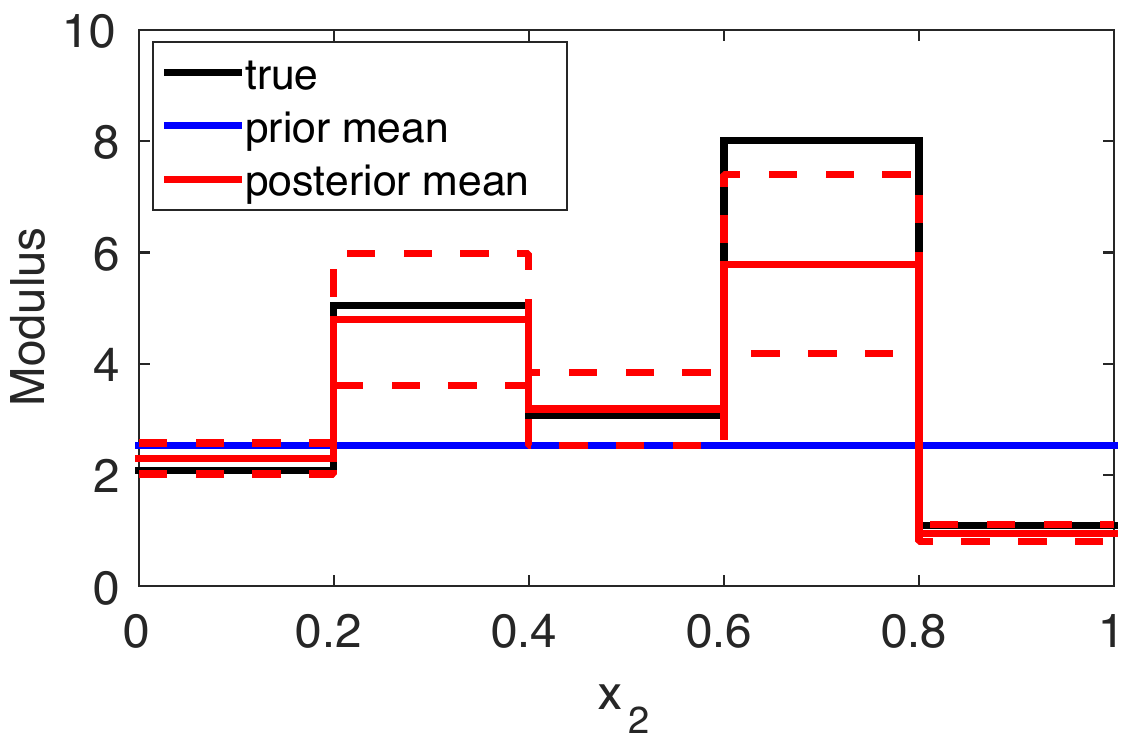}}
\caption{The mean and standard deviation of Gibbs posterior computed using data with different levels of noise for the layered material.}
\label{fig:gibbs_layered_inversion}
\end{figure}

\begin{figure}[!ht]
\centering
\subfloat[5\% noise] {\includegraphics[width=0.3\linewidth]{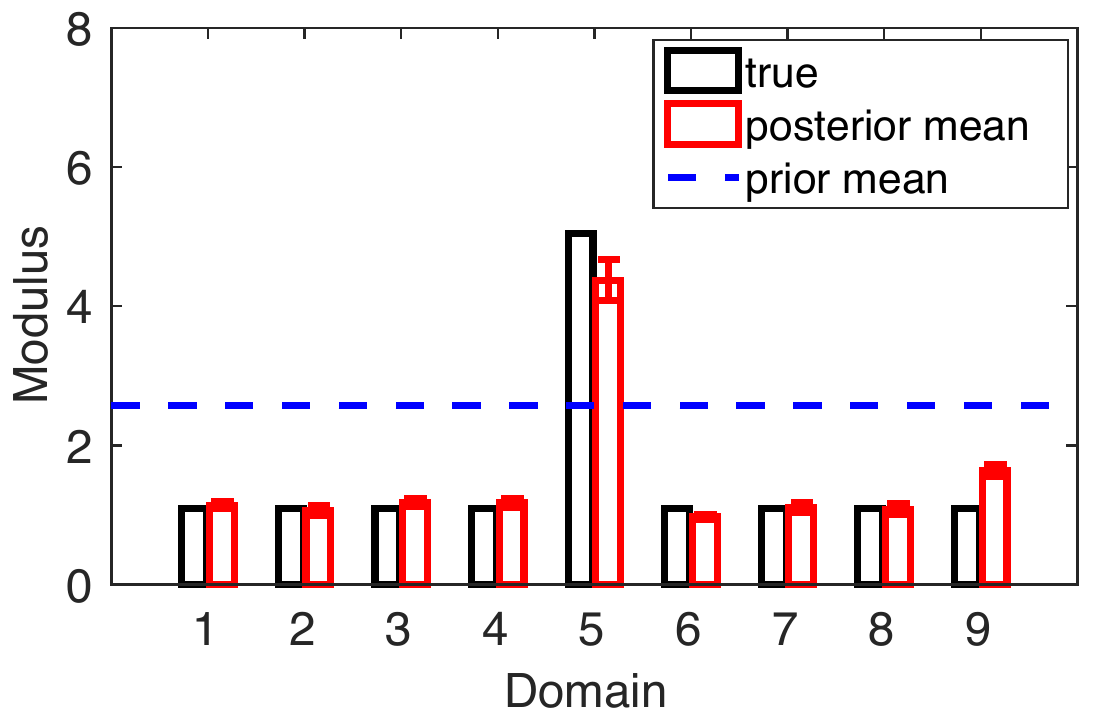}}
\hspace{0.1cm}
\subfloat[10\% noise] {\includegraphics[width=0.3\linewidth]{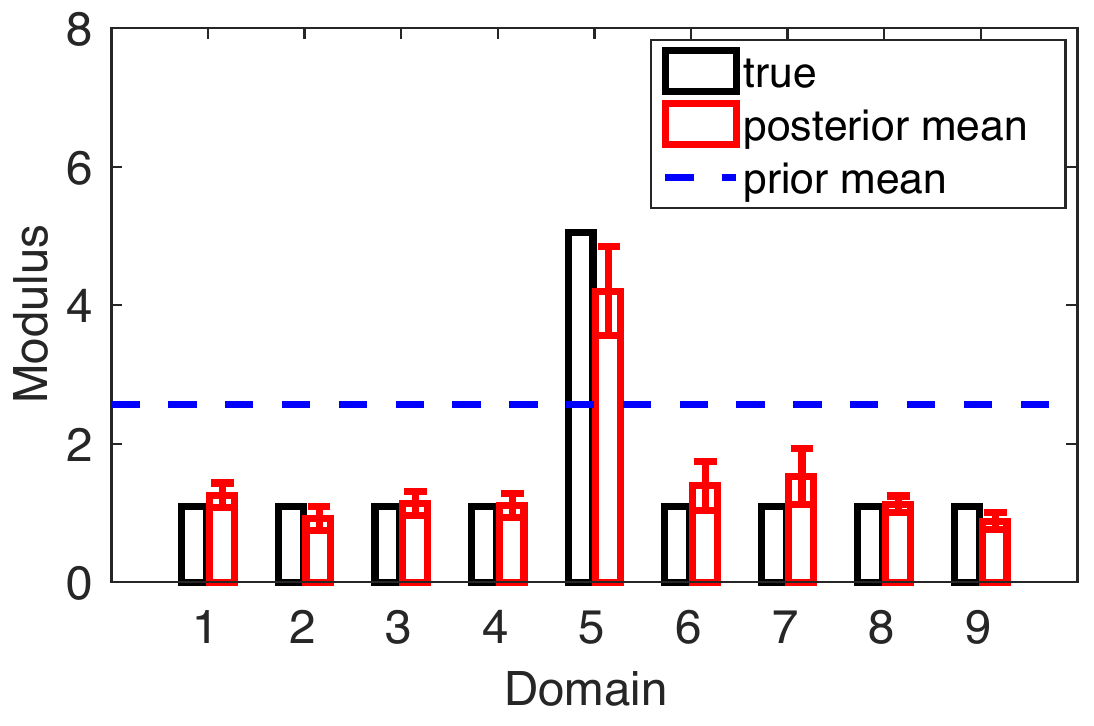}}
\hspace{0.1cm}
\subfloat[20\% noise] {\includegraphics[width=0.3\linewidth]{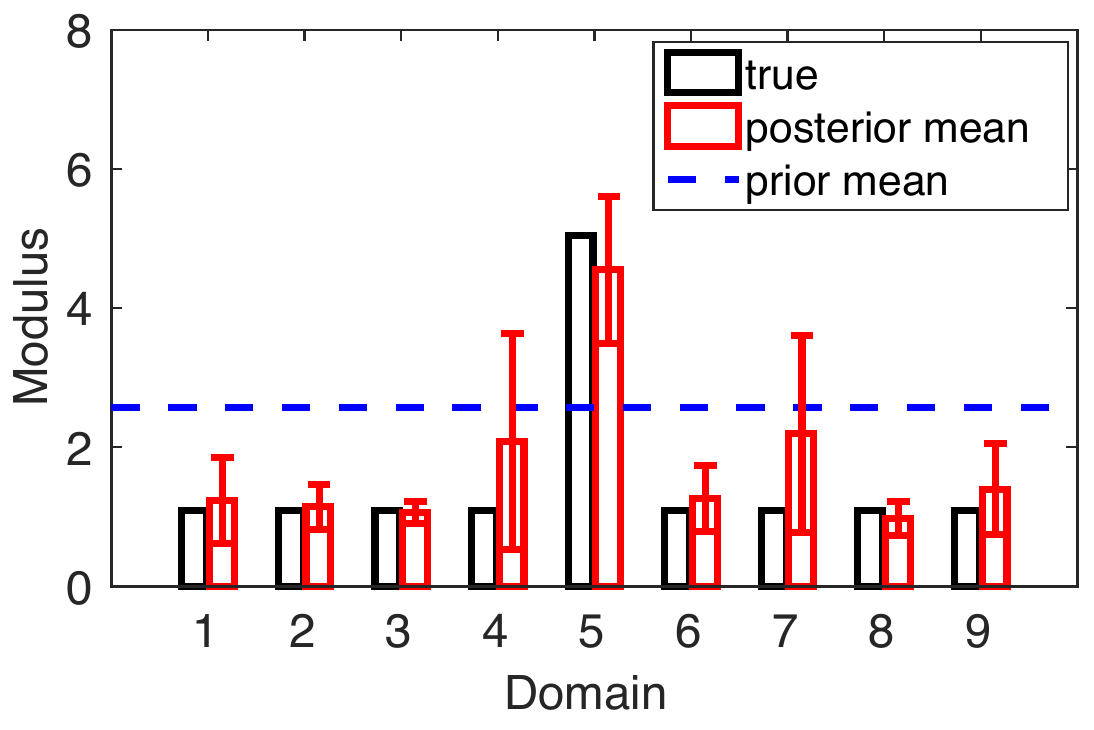}}
\caption{The mean and standard deviation of Gibbs posterior computed using data with different levels of noise for the material with a hard inclusion.}
\label{fig:gibbs_inclusion_inversion}
\end{figure}

\begin{figure}[!ht]
\centering
\subfloat[Layered material inversion] {\includegraphics[width=0.4\linewidth]{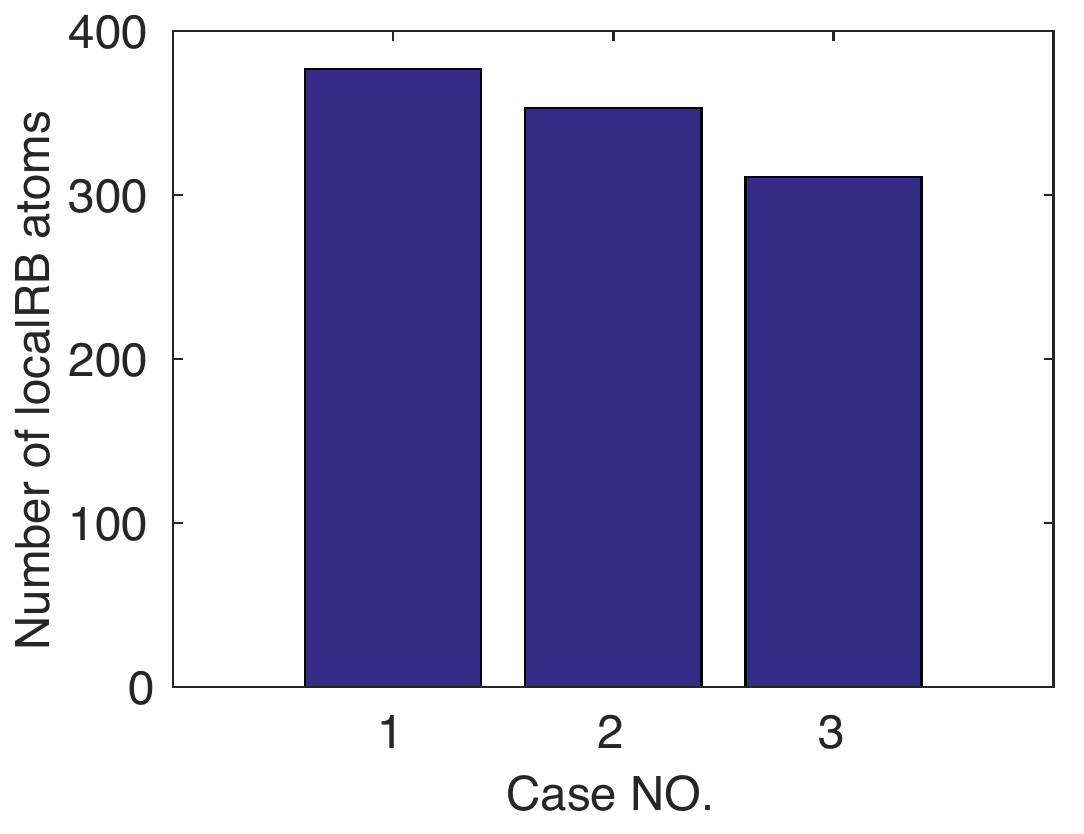}}
\hspace{0.5cm}
\subfloat[Material with a hard inclusion] {\includegraphics[width=0.39\linewidth]{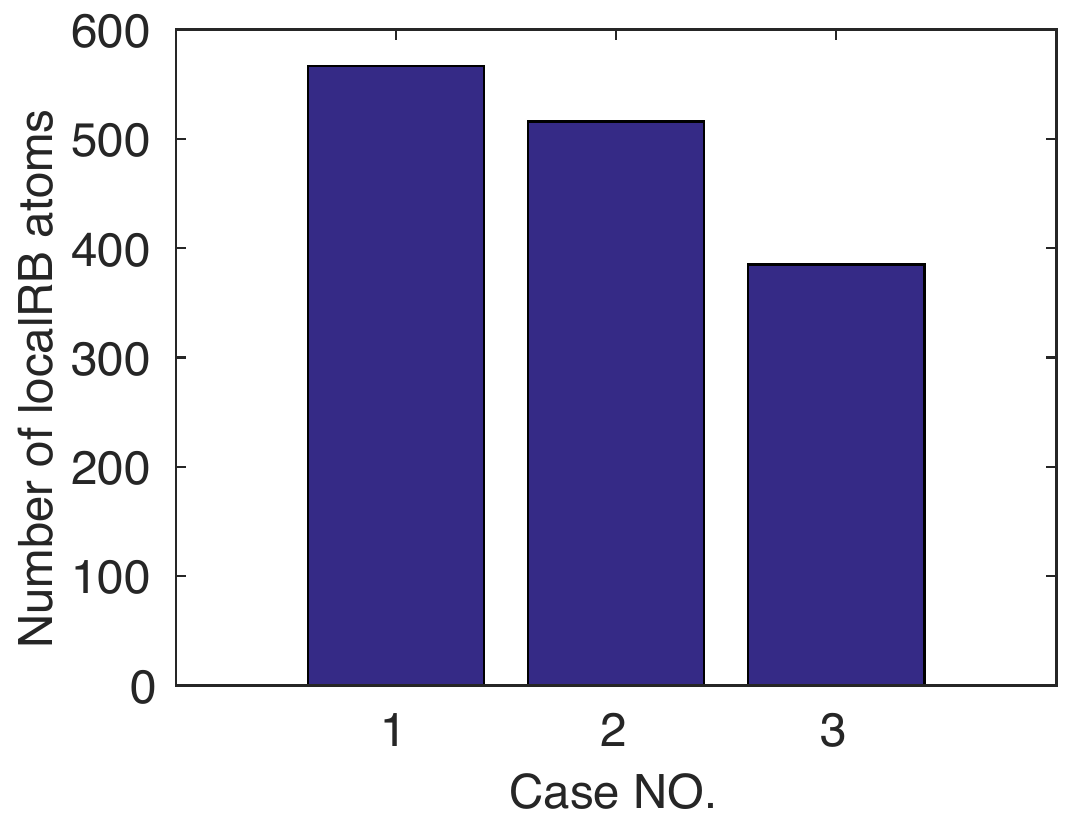}}
\caption[The total number of local RB atoms upon solving the Gibbs posterior for both material models and with different levels of noise. ]
{The total number of local RB atoms upon solving the Gibbs posterior for both material models and with different levels of noise (Case NO.).}
\label{fig:gibbs_elasticity_natom}
\end{figure}

\section{Conclusion}

In this work, we have proposed a particle-based approach with local RB surrogate model to approximate the Gibbs posterior for inverse problems. 
The Gibbs posterior has a particular advantage over the usual Bayesian approach, in the sense that it does not require an explicit model of the data generating mechanism (i.e., a likelihood function). The Gibbs posterior is applicable where the unknown parameters are connected to the data through a loss function. 
It provides a more general framework for updating belief distributions where the true data generating mechanism is unknown or difficult to specify.
We employed the local RB method to approximate the loss function in the Gibbs update formula. 
Based on a Sequential Monte Carlo (SMC) framework, we presented a method to progressive approximate the Gibbs posterior by simultaneously evolving the particles and adapting the local RB surrogate model in a sequential manner.
The emphasis of the local RB surrogate is navigated to a small fraction of the parameter space automatically by the evolving particles that progressively cluster over the support of the posterior.
Computational savings are achieved thanks to the local accuracy and the efficiency of our local RB method. 
Indeed, once the local RB surrogate becomes accurate enough (specified by a parameter representing the approximation accuracy) over the local support of the posterior, further evolution of the particles takes minimal cost.
Through several numerical examples that include advection-diffusion problems and elasticity imaging problems, we demonstrated the consistency of our method with the state-of-art Markov chain Monte Carlo (MCMC) method.
Furthermore, we showed that significant computational savings can be achieved to approximate the Gibbs posterior using our proposed method.

\bibliographystyle{plain}
\bibliography{lrb}

\end{document}